\documentclass{elsarticle}

\usepackage{amsmath,amsthm,amssymb}
\usepackage{color}
\usepackage{graphicx}
\usepackage{hyperref}
\usepackage{enumitem}
\usepackage{natbib}

\definecolor{webbrown}{rgb}{.6,0,0}
\newcommand{\CB}[1]{\textcolor{webbrown}{#1}}
\newsavebox{\tmpbox}
\newenvironment{fmpage}[1]
        {\begin{lrbox}{\tmpbox}\begin{minipage}{#1}}
        {\end{minipage}\end{lrbox}\fbox{\usebox{\tmpbox}}}

\newcommand{\condgen}[6]{{#1}#2 #5 #3 #6 #4}
\newcommand{\bbrd}[1]{\mbox{\rm{I}\kern-.1667em{#1}}}
\newcommand{\EXP}{\mathbb{E}}
\newcommand{\PROB}{\mathbb{P}}

\newcommand{\Probcsm}[2]{\condgen{\PROB}{\bigl\{}{\bigm|}{\bigr\}}{#1}{#2}}
\newcommand{\Probcmd}[2]{\condgen{\PROB}{\Bigl\{}{\Bigm|}{\Bigr\}}{#1}{#2}}

\newcommand{\Expcsm}[2]{\condgen{\EXP}{\bigl[}{\bigm|}{\bigr]}{#1}{#2}}

\newcommand{\definehere}[1]{\textcolor{webbrown}{\bf #1}}

\newtheorem{theorem}{Theorem}
\newtheorem{corollary}[theorem]{Corollary}

\newcommand{\nlik}{C}
\newcommand{\slik}{K}

\newcommand{\nolik}{B}
\newcommand{\solik}{J}

\newcommand{\nx}{u}
\newcommand{\ny}{v}
\newcommand{\nz}{w}

\newcommand{\troot}{R}

\newcommand{\fdf}{\Phi}

\begin{document}
\begin{frontmatter}
\newpageafter{abstract}

\title{Gain-loss-duplication models on a phylogeny: exact algorithms for computing the likelihood and its gradient}
\author{Mikl\'os Cs\H{u}r\"os}

\ead{csuros@iro.umontreal.ca}
\address{Department of Computer Science and Operations Research, Universit\'e de Montr\'eal; 
	C.P.~6128 succursale Centre-Ville, Montr\'eal, Qu\'ebec H3C~3J7, Canada}

\begin{highlights}
\item In evolutionary genomics, one of the most useful characteristics 
	of a gene family is its phyletic profile that gives how many copies the family has in extant genomes.
\item  Birth-and-death processes model copy number evolution by gain (lateral transfer), duplication, and loss 
within a species phylogeny.
\item We decompose the model into a probabilistic network of ancestral and conserved  
copy numbers.
\item The decomposition leads to clean, fast algorithms for ancestral
inference and model parameter optimization.
\end{highlights}

\begin{abstract}
Gene gain-loss-duplication models are commonly based on 
continuous-time birth-death processes. 
Employed in a phylogenetic context, such models have been 
increasingly popular in 
studies of gene content evolution across multiple genomes.  
While the applications are becoming more varied and demanding, 
bioinformatics methods for probabilistic inference on copy numbers 
(or integer-valued evolutionary characters, in general)
are scarce. 

We describe a flexible 
probabilistic framework for 
phylogenetic gene-loss-duplication models.
The framework is based on a novel elementary representation by 
dependent random variables 
with well-characterized 
conditional distributions: binomial, P\'olya (negative binomial), 
and Poisson.

The corresponding graphical model yields 
exact numerical procedures for  
computing the likelihood  
and the posterior distribution of ancestral copy numbers.
The resulting algorithms take quadratic time in the 
total number of 
copies. 
In addition, we show how the likelihood gradient can be computed 
by a linear-time algorithm. 
\end{abstract}


%

\end{frontmatter}

\section{Introduction}
Homology between two genes is the
equivalence relation of descent from the same ancestral 
gene~\cite{Fitch.homology}, defining the corresponding equivalence classes 
of gene {\em families}. During genome annotation,  
homologies are routinely recognized by sequence similarity, and 
annotated genes are assigned to families~\cite{RAST}.
The {\em copy number} for a family is the number of family representatives
in the genome, a non-negative integer. The {\em profile} of a family comprises the 
copy numbers across different genomes. 
Family profiles are used in evolutionary genomics to 
infer ancestral gene content~\cite{phylobd.archaea}, 
and in functional genomics to recognize associations between families~\cite{DeyMeyer}. 
Probabilistic approaches to copy number evolution are based on 
continuous-time birth-death processes~\cite{Nye.genecontent}. 
Such processes are also fundamental in queuing theory~\cite{Takacs}, 
epidemiology and population growth models~\cite{Kendall}. 
In the context of genome evolution, 
the process captures fixation events eliminating a gene (``death'')
by segmental loss and pseudogenization, 
or adding a gene (``birth'') to the genome, either by duplication
within the same genome, or
by lateral transfer from another genome. 

Bioinformatics problems for 
copy numbers differ fundamentally from molecular sequence evolution 
problems with a finite character set, and    
porting standard methods to an unbounded character domain is 
generally not possible. 
A simple, but unsatisfactory fix   
is to impose a limit on maximum copy number, and usual 
techniques like Felsenstein's peeling 
method~\cite{Felsenstein.ML.pruning} apply. 
Indeed, 
pioneer applications of birth-death processes for 
gene content evolution by Hahn et al.~\cite{Hahn2005} and by
Iwasaki and Takagi~\cite{IwasakiTakagi} employ the same workaround. 
The first algorithm for computing the 
profile likelihood for unbounded copy numbers 
with a gain-loss-duplication model
is by 
Cs\H{u}r\"os and Mikl\'os~\cite{Istvan.genecontent,phylobd.archaea}, and  
the corresponding methods are implemented in the software package 
Count~\cite{Count}, which has been used in hundreds of studies.

The algorithmics of phylogenetic birth-death models is 
difficult mainly because of unobserved empty profiles 
(gene families ``waiting to be discovered''), and of the intricate transition 
probabilities in birth-death processes without known closed expressions.  
We introduce a fresh mathematical framework for 
linear gain-loss-duplication  (and gain-loss without duplication, or duplication-loss without gain) models on a phylogeny. 
The novel formulation is based on our pivotal Theorem~\ref{tm:transition} 
giving the transition probabilities in a closed form that involves 
only basic discrete distributions.
The theorem suggests a fundamental dependency
network of random variables along the phylogeny, 
representing ancestral copy numbers and conservation. 
The elementary decomposition yields 
relatively simple algorithms to compute`
the likelihood of a family profile (Theorem~\ref{tm:lik.rec.finite}). 
While the likelihood computation algorithm
is a simpler and more insightful version of an existing method
(Theorem~\ref{tm:algo.old}), 
the network factorization 
also leads to an algorithm for posterior probabilities of ancestral copy numbers 
(Theorem~\ref{tm:olik.rec} and Corollary~\ref{cor:posterior}), 
and to our main result, a hitherto elusive
algorithm for exactly computing the gradient of 
the log-likelihood with respect to 
model parameters (Corollary~\ref{cor:loglik.d} and Theorem~\ref{tm:fd}).

\section{Theory}
A \definehere{phylogeny} is a rooted binary tree 
with nodes numbered $\nx\in [\troot]=\{1,2,\dotsc,\troot\}$.
Every node either has two non-null child nodes, or is a terminal node 
(a \definehere{leaf})
with two null children.   
For ease of notation, we assume that the nodes are indexed respecting postfix order, with
every child's index being less than the parent's, so that
the last one is the root.
The tree is identified by its root~$\troot$ and 
its edges $T\subset [\troot]\times [\troot]$ directed from parent to child. 
The edges in the subtree rooted at a node~$\nx$ are
denoted by~$T_{\nx}$, including $T_{\troot}=T$.  
The set of leaves is denoted by~$\mathcal{L}$, and 
the leaf set for~$T_{\nx}$ by $\mathcal{L}_{\nx}$; 
in particular, $\mathcal{L}=\mathcal{L}_{\troot} $.
For simplicity, 
start the indices with the leaves 
respecting the postfix order,  
so that 
$\mathcal{L}=[L]$ and every subset~$\mathcal{L}_{\nx}$ 
comprises consecutive integers.    

Consider the problem of \definehere{copy number evolution}: 
each node~$\nx$ has an associated random variable~$\xi_{\nx}$, 
called the copy number, taking non-negative 
integer values, and the joint distribution is determined 
by dependencies along the phylogeny: 
\begin{equation}\label{eq:full.likelihood.def}
\PROB\{\xi_1=n_1,\dotsc, \xi_{\troot}=n_t\} 
= \PROB\{\xi_{\troot} =n_{\troot} \} \prod_{\nx\ny\in T} \underbrace{\Probcmd{\xi_{\ny}=n_{\ny}}{\xi_{\nx}=n_{\nx}}}_{\CB{\text{transition on edge $\nx\ny$}}}.
\end{equation}
The leaf variables are observable, corresponding to extant species,
forming the \definehere{profile}
$
\Xi = \{\xi_{\ny}\}_{\ny\in \mathcal{L}}
$.
Non-leaf nodes are (hypothetical) ancestors with unobserved copy numbers. 
The \definehere{ancestral inference} problem is that of 
estimating $\{\xi_{\nx}\}_{\nx\not\in\mathcal{L}}$ for ancestral nodes, knowing the 
distribution of Eq.~\eqref{eq:full.likelihood.def} and~$\Xi$.  

Suppose that we observe the leaf variables across 
sample profiles called \definehere{families} $f=1,\dotsc,F$, 
with independent and identically distributed (\definehere{iid}) 
copy number vectors $(\xi_{1,1},\dotsc,\xi_{1,R}),
\dotsc,(\xi_{F,1},\dotsc,\xi_{F,R})$.  
The \definehere{model inference} problem is that of deducing   
the distribution of~\eqref{eq:full.likelihood.def}
from an iid sample $(\Xi_1,\dotsc, \Xi_F)$.

A linear birth-death model defines the probabilities 
$\Probcmd{\xi_{\ny}=m}{\xi_{\nx}=n}$ 
along every edge~$\nx\ny$ by a 
continuous-time Markov process
${\bigl\{\xi(t)\colon 0\le t\le t_{\nx\ny}\bigr\}}$ via 
$\xi_{\nx}=\xi(0)$ and $\xi_{\ny}=\xi(t_{\nx\ny})$ 
during some time $t_{\nx\ny}\ge 0$ (the \definehere{edge length}).
The process is characterized by the
constant instantaneous rates for loss $\mu>0$,
duplication $\lambda\ge 0$ and $\kappa\ge 0$, 
so that $n\to (n-1)$ death events arrive with a rate of $\mu n$, 
and $n\to (n+1)$ birth events arrive either 
with a rate of $\lambda (n+\kappa)$. 
In particular, 
for $p_n(t)=\PROB\{\xi(t)=n\}$, 
the Kolmogorov backward equations are 
\[
p_n'(t)
  =  \{n>0\} \lambda (n-1+\kappa) p_{n-1}(t)+\mu (n+1) p_{n+1}(t)
 	-(\lambda (n+\kappa) +\mu n) p_n(t)
\]
with $p_n'(t)=\frac{\partial p_n(t)}{\partial t}$. 
The no-duplication model 
is the limit for $\kappa\lambda\to\gamma\mu$ while $\lambda\to 0$: 
\[
p_n'(t)
  =  \{n>0\} \mu \gamma p_{n-1}(t)+\mu (n+1) p_{n+1}(t)
 	-\mu (\gamma + n) p_n(t),
\]
characterized by loss rate~$\mu$ and the relative gain rate~~$\gamma$. 
The model's rate parameters have convenient biological interpretations.  
The components~$\mu$ and~$\lambda$ are the per-copy 
instantaneous rates of loss and duplication. 
The $\kappa$ and $\gamma$ parameters represent the propensity for gene acquisition from 
external sources, acting as an {\em environmental} fractional copy that 
contributes $\kappa\lambda$ (if $\lambda>0$) or $\gamma\mu$ (if $\lambda=0$) to the 
birth rate. In other words, the constant {\em gain} rate component serves as an abstraction of horizontal gene transfer 
from all sources in the organism's environment. 
In constrast, {\em duplication} originates from 
copies within the genome, each contributing~$\lambda$ to the birth rate.  
Duplication-loss rates can be embedded in a population-genetic model
of genome size evolution, 
so that they are determined by a family-specific selection coefficient, 
and (constant) population size \cite{SelaWolfKoonin}. 
The particular case $\mu,\lambda>0$ and $\kappa=1$ is 
the process of gene length evolution in the Thorne-Kishino-Felsenstein model~\cite{TKF},
where the {\em immortal link} plays the same role as the environmental copy here.   

The gene copies evolve independently, forming a set of 
Galton-Watson trees on each edge~\cite{Istvan.genecontent}.
(Every copy at the ancestor and every gained copy defines the root
of a phylogeny over the copies with time-annotated nodes;  
loss events create terminal nodes and duplication events create bifurcations.)  
Our interest lies not in inferring the 
trees (i.e., in {\em reconciling} the gene histories with the species phylogeny), 
but rather in designing a model 
for the copy numbers without explaining about each 
copy where they originate.

The transition probabilities for an arbitrary starting value 
$\xi(0)=n$ are given in Theorem~\ref{tm:transition}.
The basic transition probabilities are well understood~\cite{Kendall,KarlinMcgregor}:
\begin{subequations}\label{eq:transition.basic}
\begin{align}
h_n(t) = 
\Probcmd{\xi(t)=n}{\xi(0)=0} 
	& = \binom{\kappa+n-1}{n} (1-q)^{\kappa}q^n \label{eq:polya}
	\quad \text{if $\lambda,\kappa>0$}
\\
h_n(t) = \Probcmd{\xi(t)=n}{\xi(0)=0} \label{eq:poisson}
	& = e^{-r}\frac{r^n}{n!}
	\quad\text{if $\lambda=0, \kappa>0$}
\\
g_n(t) = \Probcmd{\xi(t)=n}{\xi(0)=1}
	& = \begin{cases}
		p & \{n=0\}\\
		(1-p)(1-q)q^{n-1} & \{ n>0\}
	 	\end{cases}
	 	\quad\text{if $\kappa=0, \lambda>0$}\notag
\end{align}
\end{subequations}
with the parameters 
\begin{subequations}\label{eq:param.probs}
\begin{align}
p & = \frac{\mu-\mu e^{-(\mu-\lambda)t}}{\mu-\lambda e^{-(\mu-\lambda)t}} 
\\
q & = \frac{\lambda-\lambda e^{-(\mu-\lambda)t}}{\mu-\lambda e^{-(\mu-\lambda)t}}
\quad \text{ if $\lambda>0$}
&
r & = \gamma (1-e^{-\mu t}) \quad \text{ if $\lambda=0$}
\intertext{assuming $\lambda \ne \mu$; or if $\lambda=\mu$,} 
p & = q= \frac{\mu t}{1+\mu t}.
\end{align}
\end{subequations}
The P\'olya distribution of~\eqref{eq:polya} is the generalized version of 
the negative binomial, allowing for non-integer~$\kappa$ parameter.
Recall that the generalized binomial coefficient~$\binom{\theta}{k}$ 
 for all $\theta\in\mathbb{R}$
and nonnegative integer $k\in\mathbb{N}$ is defined by 
\[
\binom{\theta}{k} = \frac{(\theta)_k}{(k)_k}
\quad\text{with} \quad
(\theta)_k = \begin{cases} 
1 & \{k=0\}\\
\theta\times (\theta-1)_{k-1} = 
\theta (\theta-1)\dotsm (\theta-k+1) & \{k>0\}
\end{cases}
\]
So, the point mass function for the P{\'o}lya distribution with parameters $(\kappa, q)$ 
is 
\[
h_n(t) = \begin{cases}
	(1-q)^\kappa & \{n=0\}\\
	\frac{\kappa(\kappa+1)\dotsm (\kappa+n-1)}{n!} (1-q)^\kappa q^n & \{n>0\} 
	\end{cases}
\]
The rates and the edge length can be rescaled simultaneously without 
affecting the distributions. Dissecting into scale-independent parameters
(assuming $q\ne p$):
\begin{align*} 
p & = \frac{1-e^{-\delta(\mu t)}}{1-(1-\delta)e^{-\delta(\mu t)}} & 1-p & = \frac{\delta e^{-\delta (\mu t)}}{1-(1-\delta)e^{-\delta(\mu t)}}\\
q & = \frac{(1-\delta)\Bigl(1-e^{-\delta(\mu t)}\Bigr)}{1-(1-\delta)e^{-\delta(\mu t)}}
	& 1- q & = \frac{\delta}{1-(1-\delta)e^{-\delta(\mu t)}}
\end{align*}
with $\delta = 1-\frac{\lambda}{\mu} = 1-\frac{q}{p}$. 
The formulas are invertible: for a given $0< p,q<1$ we can find 
$\delta$ and the scaled edge length $(\mu t)$. 
\begin{theorem}[Unicity of distribution parameters]\label{tm:param.probs}
Let $0<t$ be fixed. For any given $0<p,q<1$ and $0<\kappa$, 
or with $q=0$, for any given $0<p<1$ and $0<r$,
there exist valid rate settings $0<\mu, 0\le \lambda$  
that yield those distribution parameters as in Eq.~\eqref{eq:param.probs} 
\end{theorem}
\begin{proof}
If $q=0$, then $\lambda=0$, and by $p=1-e^{-\mu t}$ and $r=\gamma p$, 
we can set $\mu t$ and $\gamma$ to match~$p$ and~$r$.   
If $0<q=p$, then set $(\mu t)=p/(1+p)$ and $\lambda=\mu$.  
Otherwise, since $q/p = 1-\delta$ and $(1-q)/(1-p) = e^{\delta\mu t}$, 
set $\delta = 1-\frac{q}{p}$,  
$(\mu t) = \frac{\ln \frac{1-q}{1-p}}{1-\frac{q}{p}}$ and $\lambda = \mu (1-\delta)$.
\end{proof}
Note that even if the birth-death process has a stationary distribution only when 
$\lambda \le \mu$ or $\delta\ge 0$, the formulas remain 
valid for all transient probabilities ($t<\infty$) 
even when $\lambda > \mu$. 
\section{Results and discussion}
\subsection{Transient probabilities in the general case}
First, suppose that duplications are allowed, and $\lambda_{\ny}>0$ on all edges $\nx\ny\in T$. 
If there are $\xi_{\nx}=n$ copies 
at an ancestral node~$\nx$, then they evolve independently along
each child edge~$\nx\ny$:
\begin{equation}\label{eq:groups.rv}
\xi_{\ny} = \zeta_0 + \zeta_1 + \dotsm + \zeta_n
\end{equation}
where $\zeta_0$ denotes 
the {\em xenolog copies}, 
and $\zeta_i$ denote iid variables for the descendant 
{\em inparalog} copies 
from each ancestral instance $i=1,\dotsc, n$.  
The $\zeta_i$ variables 
follow the basic transition probabilities 
\[
\PROB\{\zeta_0=k\} = h_k(t_{\nx\ny})
\quad\text{and}\quad
\PROB\{\zeta_i=k\} = g_k(t_{\nx\ny})
	\quad \text{for all $i>0$.}
\]
The key observation for calculating 
$\Probcmd{\xi_{\ny}=m}{\xi_{\nx}=n} = \PROB\{\zeta_0+\zeta_1+\dotsm + \zeta_n=m\}$
is that $\zeta_i|\zeta_i>0$ 
has the same geometric tail as the P{\'o}lya distribution 
of $\zeta_0$.  
Since the distributions with the same tail parameter can be summed
at ease, $\xi_{\ny}-s$ has a P{\'o}lya distribution 
with parameter $(\kappa+s)$, where~$s=\sum_{i=1}^n\{\zeta_i>0\}$ is the 
number of conserved copies.
(The shorthand notation $\{\zeta_i>0\}$ denotes indicator variable  that takes 
the value~1
whenever $\zeta_i$ is positive, and the value~0 when $\zeta_i=0$.)

\begin{theorem}[Transient probabilities in the general case]\label{tm:transition}
For a linear birth-death process with parameters $\kappa,\lambda,\mu>0$, 
\begin{multline}
\Probcmd{\xi(t)=m}{\xi(0)=n} 
\\
	= \sum_{s=0}^{\min\{n,m\}}
		\binom{\kappa + m-1}{m-s} (1-q)^{\kappa+s} q^{m-s} \binom{n}{s}p^{n-s}(1-p)^s
\end{multline}
with the parameters $p,q$ defined in Eqs.~\eqref{eq:param.probs}.

For a linear birth-death process with parameters $\lambda=0$ and $\mu,\gamma>0$, 
\begin{equation}
\Probcmd{\xi(t)=m}{\xi(0)=n} 
	= \sum_{s=0}^{\min\{n,m\}}
		e^{-r} \frac{r^{m-s}}{(m-s)!} \binom{n}{s} (1-p)^s p^{n-s}
\end{equation}
with the parameters $p,r$ defined in Eqs.~\eqref{eq:param.probs}.
\end{theorem}

\subsection{Phylogenetic model with conserved copies}
We amend the phylogenetic model by 
explicitly inserting 
a hidden random variable~$\eta_{\ny}$ 
of conserved ancestral copies between 
the copy numbers~$\xi_{\nx}$ and~$\xi_{\ny}$ on every edge $\nx\ny$.  
For the ease of presentation, we continue 
with $\lambda_{\ny}>0$ at every node~$\ny$, 
and return to the no-duplication model afterwards.  
Using Theorem~\ref{tm:transition}, 
\begin{subequations}
\begin{align}
\Probcmd{\eta_{\ny}=s}{\xi_{\nx}=n} & = \binom{n}{s}(1-p_{\ny})^s(p_{\ny})^{n-s} 
& \{s\le n\} \label{eq:cdist.edge}
\\
\Probcmd{\xi_{\ny}=m}{\eta_{\ny}=s} & = \binom{\kappa_{\ny} + m-1}{m-s} (1-q_{\ny})^{\kappa_{\ny}+s} (q_{\ny})^{m-s} 
& \{s\le m\} \label{eq:cdist.node}
\end{align}
\end{subequations}
with edge-specific loss, duplication, and gain parameters $p_{\ny},q_{\ny},\kappa_{\ny}$.  
%
%
%
%
%
%
%

A complete history 
fixes all counts $\xi_{\nx}$ and $\eta_{\nx}$:
$\{\xi_1=n_1,\dotsc,\xi_{\troot}=n_{\troot}, \eta_1=s_1,\dotsc, \eta_{R-1}=s_{R-1}\}$.
The joint distribution of our 
phylogenetically linked random variables 
is written explicitly as  
\begin{multline}\label{eq:likelihood.history}
\PROB\{\xi_1=n_1,\dotsc,\xi_{\troot}=n_{\troot}, \eta_1=s_1,\dotsc, \eta_{R-1}=s_{R-1}\}
\\
 = 
 \begin{aligned}[t]
 \PROB\{\xi_{\troot}=n_{\troot}\}
 \times 
 \prod_{\nx\ny\in T}
 	\biggl( 
 	& \underbrace{\binom{n_{\nx}}{s_{\ny}} (1-p_{\ny})^{s_{\ny}} (p_{\ny})^{n_{\nx}-s_{\ny}}}_{\CB{\Probcsm{\eta_{\ny}=s_{\ny}}{\xi_{\nx}=n_{\nx}}}}
 	 \\ & \times 
 	\underbrace{\binom{\kappa_{\ny}+n_{\ny}-1}{n_{\ny}-s_{\ny}} (1-q_{\ny})^{\kappa_{\ny}+s_{\ny}} (q_{\ny})^{n_{\ny}-s_{\ny}}}_{\CB{\Probcsm{\xi_{\ny}=n_{\ny}}{\eta_{\ny}=s_{\ny}}}}
 	\biggr),
 \end{aligned}
\end{multline}
All histories satisfying $s_{\ny}\le \min\{n_{\nx},n_{\ny}\}$ 
on every edge~$\nx\ny\in T$ and $\PROB\{\xi_{\troot}=n_{\troot}\}\ne 0$
have positive probability if $p_{\nx},q_{\nx}$ are bounded away from~0 and~1.   

Let $\Xi=\{n_{\ny}\colon \ny\in \mathcal{L}\}$ be a profile 
comprising the observed copy numbers. 
The \definehere{profile likelihood} 
is the sum of all history probabilities 
	from~\eqref{eq:likelihood.history} for the same profile:
\[
L(\Xi)= \PROB\{\Xi\} = \sum_{n_{\nx},s_{\nx}\colon x\not\in\mathcal{L}}
	\PROB\{\xi_1=n_1,\dotsc,\xi_{\troot}=n_{\troot}, \eta_1=s_1,\dotsc, \eta_{R-1}=s_{R-1}\},
\]
with infinitely many terms. 
Define the \definehere{partial profile} within every subtree as
$
\Xi_{\nx} = \{\forall \ny\in \mathcal{L}_{\nx} \colon \xi_{\ny}=n_{\ny}\}
$
where~$\mathcal{L}_{\nx}$ denotes the leaves in the subtree rooted 
at~$\nx$, including the singleton $\mathcal{L}_{\nx}=\{\nx\}$ 
whenever~$\nx$ is a leaf. 
Define the  likelihood of the partial profiles
conditioned on $\xi_{\nx}$ or $\eta_{\nx}$:
\[
\nlik_{\nx}(n) = \Probcmd{\Xi_{\nx}}{\xi_{\nx}=n}
\qquad\text{and}\qquad 
\slik_{\nx}(s) = \Probcmd{\Xi_{\nx}}{\eta_{\nx}=s}.
\]
At a leaf~$\nx$, we have $\nlik_{\nx}(n)=1$ if $n=n_{\nx}$, the observed count, or 
$\nlik_{\nx}(n)=0$ if $n\ne n_{\nx}$.  
All other conditional likelihoods can be expressed using Equations~\eqref{eq:cdist.edge}
and~\eqref{eq:cdist.node} about the conditional distributions $\xi_{\nx}\mid \eta_{\nx}$ 
and $\eta_{\ny}\mid \xi_{\nx}$. 
At all nodes~$\nx$, 
\begin{subequations}\label{eq:lik.rec.infty} 
\begin{align}
\slik_{\nx}(s) & = \sum_{k=0}^{\infty} \binom{\kappa_{\nx}+s+(k-1)}{k}(1-q_{\nx})^{\kappa_{\nx}+s}(q_{\nx})^k\times \nlik_{\nx}(s+k);
\intertext{and at every ancestral node~$\nx$,} 
\nlik_{\nx}(n) 
		& = \prod_{\nx\ny\in T}
		 \biggl(\sum_{s=0}^n \binom{n}{s} (1-p_{\ny})^{s} (p_{\ny})^{n-s} \times \slik_{\ny}(s)\biggr).
\end{align}
\end{subequations}

The family distribution at the root~$R$ is needed to sum across 
the likelihoods $\nlik_{\troot}(n)$ to get the profile likelihood
\[
L(\Xi) = \PROB\{\Xi\} = \sum_{n=0}^{\infty} \PROB\{\xi_{\troot}=n\} \Probcsm{\Xi_{\troot}}{\xi_{\troot}=n}
= \sum_{n=0}^{\infty} \PROB\{\xi_{\troot}=n\} \nlik_{R}(n).
\] 
Assume that the root copy number follows a P{\'o}lya distribution
with some parameters $\kappa_{\troot},q_{\troot}>0$: 
\begin{equation}\label{eq:root.prior}
L(\Xi)  = \binom{\kappa_{\troot}+n-1}{n} (1-q_{\troot})^r (q_{\troot})^n \nlik_{\troot}(n). 
\end{equation}
After defining  
$\eta_{\troot}=0$, Eq.~\eqref{eq:root.prior} is the same formula 
for the likelihoods~$\slik_{\troot}$ as on the edges, 
and $L(\Xi)=\slik_{\troot}(0)$.

\subsection{Empty profile likelihood}
Typically, the input sample does not include families with an
\definehere{empty profile} that has $\xi_{\ny}=0$ at all leaves~$\ny$.
The model defines the probability of such a profile. 

\begin{theorem}[Empty profile likelihood]\label{tm:empty}
Define  $\epsilon_{\nx}=0$ for all leaves $\nx$, and 
for every non-leaf~$\nx$, 
$\epsilon_{\nx}= \prod_{\nx\ny\in T}\tilde{p}_{\ny}$ 
with \[
\tilde{p}_{\ny}=\bigl(p_{\ny}+(1-p_{\ny})\epsilon_{\ny}(1-\tilde{q}_{\ny})\bigr)
=\frac{p_{\ny}(1-\epsilon_{\ny})+\epsilon_{\ny}(1-q_{\ny})}{1-q_{\ny}\epsilon_{\ny}}
\]
at every non-root~$\ny$, and 
\[
\tilde{q}_{\nx} = q_{\nx}\frac{1-\epsilon_{\nx}}{1-q_{\nx}\epsilon_{\nx}}
\]
at every node $\nx$.
The 
probability of the empty profile is 
\[
L(0) = \prod_{\nx=1}^R(1-\tilde{q}_{\troot})^{\kappa_{\troot}}
	= \prod_{\nx=1}^R \biggl(\frac{1-q_{\nx}}{1-q_{\nx}\epsilon_{\nx}}\biggr)^{\kappa_{\nx}}.
\]
\end{theorem}

Let the input sample consist of the observed profiles 
for families $f=1,\dotsc,F$: 
$\Xi_f=\Bigl\{\xi_{\nx}=n_{f,{\nx}}\colon \nx\in \mathcal{L}\bigr\}$. 
If the empty profiles are unobservable, then the likelihood of a single family profile is 
conditioned on the fact that at least one copy number is positive:
\begin{align*}
L^*(\Xi_f) & 
	= \Probcmd{\forall \nx\in \mathcal{L}_{\troot}\colon \xi_{\nx}=n_{f,\nx}}{\exists \nx\in \mathcal{L}_{\troot}\colon \xi_{\nx}\ne 0}
	= \frac{L(\Xi_f)}{1-L(0)},
\end{align*}
using the uncorrected likelihoods~$L(\Xi)$
without conditioning on being empty, and in particular 
the empty profile likelihood~$L(0)$ from Theorem~\ref{tm:empty}.
Applying the correction to the entire sample: 
\begin{equation}\label{eq:lik.corr}
L^* = \prod_{f=1}^F L^*(\Xi_f) = \frac{\prod_{f=1}^F L(\Xi_f)}{(1-L(0))^{F}}.
\end{equation}
The correction of Equation~\eqref{eq:lik.corr} 
is akin to Felsenstein's 
likelihood correction formula for restriction site evolution~\cite{Felsenstein.restml}.


\subsection{Computing the profile likelihood}
Since the ancestors' copy number $\{\xi_{\nx}=n\}$ may be possible for all nonnegative~integers~$n$,
the likelihood recurrences of~\eqref{eq:lik.rec.infty}
involve infinite sums for~$\slik_{\nx}$, and 
infinitely many~$\nlik_{\nx}(n)$. 
We can, however, factor out the histories
with parallel losses for a finite calculation.
Define $\tilde{\xi}_{\nx}$ at every ancestral node~${\nx}$  
as the number of copies that are not lost simultaneously in
all descendant lineages to~$\mathcal{L}_{\nx}$.  
Let~$\tilde{\eta}_{\nx}$ denote 
the number of ancestral copies that are not lost either 
on the edge leading to~$\nx$ or in the subtree~$T_{\nx}$. 
In other words, $\tilde{\eta}$ and $\tilde{\eta}$ count 
only the progenitors of copies at the leaves.
(Note that the {\em ancestral} copy numbers $\tilde{\xi},\tilde{\eta}$ count 
the ancestral genes of extant copies, as opposed 
to the {\em ancestors'} copy numbers $\xi,\eta$ that count 
all homologs in the ancestors' genomes.) 
Define~$\epsilon_{\nx}$, $\tilde{p}_{\ny}$ and $\tilde{q}_{\nx}$ as in 
Theorem~\ref{tm:empty}.
Since 
ancestral copies are lost independently with probability~$\epsilon_{\nx}$,  
for $0\le \ell\le n$, 
$\Probcsm{\tilde{\xi}_{\nx}=\ell}{\xi_{\nx}=n}
	= \binom{n}{\ell}
		(1-\epsilon_{\nx})^{\ell} (\epsilon_{\nx})^{n-\ell}$
and, for all $0\le s\le t$,
$\Probcsm{\tilde{\eta}_{\nx}=s}{\eta_{\nx}=t}
	= \binom{t}{s}
		(1-\epsilon_{\nx})^{s} (\epsilon_{\nx})^{t-s}
$.

\begin{theorem}[Likelihood computation]\label{tm:lik.rec.finite}
Given a profile~$\Xi$, 
define the conditional likelihoods 
\[
\tilde{\slik}_{\nx}(s) = \Probcsm{\Xi_{\nx}}{\tilde{\eta}_{\nx}=s}
\quad\text{and}\quad 
\tilde{\nlik}_{\nx}(\ell) = \Probcmd{\Xi_{\nx}}{\tilde{\xi}_{\nx}=\ell}
\]
at all nodes~$\nx$. In particular, 
the profile likelihood is $L(\Xi)=\tilde{\slik}_{\troot} (0)$ 
at the root~$R$. Define the sum of observed leaf copy numbers 
within every subtree:
$
m_{\nx} = \sum_{\ny\in \mathcal{L}_{\nx}} n_{\ny}
$.
\begin{enumerate}[label=(\roman*)]
\item  For all $s>m_{\nx}$, $\tilde{\slik}_{\nx}(s)=0$ 
 and for all $\ell>m_{\nx}$, $\tilde{\nlik}_{\nx}(\ell)=0$. 
\item  At every node~$\nx$, for all $0\le s\le m_{\nx}$, 
\begin{equation}\label{eq:lik.rec.edge}
\tilde{\slik}_{\nx}(s) 
= \sum_{\ell= s}^{m_{\nx}}
	\tilde{\nlik}_{\nx}(\ell)	
	\times
	\binom{\kappa_{\nx}+\ell-1}{\ell-s} (1-\tilde{q}_{\nx})^{\kappa_{\nx}+s} (\tilde{q}_{\nx})^{\ell-s}.
\end{equation}
\item If $\nx$ is a leaf, then $\tilde{\nlik}_{\nx}(\ell) = \{\ell=n_{\nx}\}$. 
If~$\nx$ is an ancestral node with children $\nx\ny,\nx\nz\in T$, then 
for all $0\le \ell \le m_{\nx}=m_{\ny}+m_{\nz}$, 
\begin{multline}\label{eq:lik.rec.node}
\tilde{\nlik}_{\nx}(\ell)
 = \sum_{s=0}^{\min\{\ell,m_{\ny}\}} 
 		\tilde{\slik}_{\ny}(s) 
 		\times \tilde{\slik}^{\ell}_{\nz}(\ell-s)
 		\\
 		\times \binom{\ell}{s} \Bigl(\frac{1-\tilde{p}_{\ny}}{1-\tilde{p}_{\ny}\tilde{p}_{\nz}}\Bigr)^s
 		 \Bigl(\frac{\tilde{p}_{\ny} -\tilde{p}_{\ny} \tilde{p}_{\nz}}{1-\tilde{p}_{\ny}\tilde{p}_{\nz}}\Bigr)^{\ell-s} 		
\end{multline}
with $
\tilde{\slik}_{\nz}^{\ell}(\ell)=\tilde{\slik}_{\nz}(\ell)
$, and, for all $0\le d< \ell$,  
\begin{equation}
\tilde{\slik}^{\ell}_{\nz}(d)= (1-\tilde{p}_{\nz}) \tilde{\slik}^{\ell}_{\nz}(d+1) 
		+ \tilde{p}_{\nz} \tilde{\slik}^{\ell-1}_{\nz}(d). 
\end{equation}
\end{enumerate}
\end{theorem}
Note that Equation~\eqref{eq:lik.rec.edge} also applies to 
a duplication-loss ($\lambda_{\nx},\mu_{\nx}>0$) model with no gain($\kappa_{\nx}=0$). Since
$\tilde{\xi}_{\nx}$ is the sum of $s=\tilde{\eta}_{\nx}$ geometric distributions, 
it has a negative binomial distribution with parameters $s$ and $\tilde{q}$. 
So,  
$\tilde{\slik}_{\nx}(0) = \tilde{\nlik}_{\nx}(0)$, and for all
$1\le s\le m_{\nx}$,
\[
\tilde{\slik}_{\nx}(s) = \sum_{\ell=0}^{m_{\nx}}
	\binom{\ell-1}{\ell-s}
		(1-\tilde{q}_{\nx})^{s} \tilde{q}_{\nx}^{\ell-s}. 
\]

\subsection{Multifurcations, missing data and partial genomes}
The graphical model, as presented, assumes~(1)~a binary phylogeny, 
(2)~unambiguous observation of the copy numbers~$\xi_{\nx}$ at the leaves, 
and~(3) a complete annotated genome. All three assumptions can be relaxed.   
\subsubsection*{Non-binary phylogeny}
A {\em degenerate} phylogeny~$T$ represents the parent-child 
relationships in a non-binary rooted tree. 
In such a phylogeny, the ancestral nodes may have 2 or more children. In practice, 
it makes sense to put multifurcating nodes at deep ancestors to represent the 
ambiguity of resolving short edges, and a ternary root is common
if the phylogeny was derived from an unrooted tree.
The likelihood recurrences of Theorem~\ref{tm:lik.rec.finite} can accommodate 
any $d$-ary node, by considering survival in $1,2,3,\dotsc d$ 
child lineages incrementally (for any child ordering). 
\begin{theorem}[Likelihood recurrences for multifurcating node]\label{tm:lik.rec.multi.right}
Let~$\nx$ be a node in a degenerate phylogeny with $d\ge 2$ distinct children 
$\nx\ny_1,\dotsc, \nx\ny_d\in T$ enumerated in any order.
Let 
$\epsilon_{u,-i} = \prod_{j=i}^d \tilde{p}_{\ny_j}$
for $i=1,\dotsc,d$, so that $\epsilon_{\nx} = \epsilon_{\nx,-1}$.
Define the likelihoods $\tilde{\nlik}_{\nx}^{-i}(\ell)$ and $\tilde{\slik}_{\ny_{i..d}}^{\ell}(s)$
conditioned on~$s$ surviving copies 
in the subtrees of $\ny_i,\ny_{i+1},\dotsc,\ny_d$:
\begin{subequations}\label{eq:lik.rec.multi.right}
\begin{align}
\tilde{\nlik}_{\nx}^{-d}(\ell) & = \slik_{\ny_d}(\ell) \qquad \{ 0\le \ell\le m_{\ny_d}\}
\intertext{and, for all $0< i\le d$ and for all $0\le \ell \le m_{\ny_{i-1}}+\dotsb+m_{\ny_{d}}$}
\tilde{\nlik}_{\nx}^{-(i-1)}(\ell) & = \label{eq:lik.rec.multi.node}
	\begin{aligned}[t]
	\sum_{s=0}^{\min\{\ell,m_{\ny_{i-1}}\}} 
		& 
		\tilde{\slik}_{\ny_{i-1}}(s) 
			\times \tilde{\slik}_{\ny_{i..d}}^{\ell}(\ell-s)
		\\	
		\times & \binom{\ell}{s} \biggl(\frac{1-\tilde{p}_{\ny_{i-1}}}{1-\epsilon_{\nx,-(i-1)}}\biggr)^s
		\biggl(\frac{\tilde{p}_{\ny_{i-1}}-\tilde{p}_{\ny_{i-1}}\epsilon_{\nx,-i}}{1-\epsilon_{\nx,-(i-1)}}\biggr)^{\ell-s},
		\end{aligned} 
\intertext{with}
\tilde{\slik}_{\ny_{i..d}}^{\ell}(\ell) & = \tilde{\nlik}_{\nx}^{-i}(\ell)
\\
\tilde{\slik}_{\ny_{i..d}}^{\ell}(k) & = (1-\epsilon_{\nx,-i}) \tilde{\slik}_{\ny_{i..d}}^{\ell}(k+1)
	+ \epsilon_{\nx,-i} \tilde{\slik}_{\ny_{i..d}}^{\ell-1}(k)
		\quad \{ 0\le k<\ell\}
\end{align}
\end{subequations}
Then $
\tilde{\nlik}_{\nx}(\ell) = \tilde{\nlik}_{\nx}^{-1}(\ell)$.
\end{theorem}
\subsubsection*{Missing copy numbers}
Instead of setting~0, a copy number can be declared unobserved, or ambiguous. In the likelihood  
computations, such a profile can be accommodated by imposing $\tilde{\slik}_{\nx}(s)=1$ and 
$\tilde{\nlik}_{\nx}(\ell)=1$ at all nodes~$\nx$ for which the 
partial profile is ambiguous at every leaf~$\mathcal{L}_{\nx}$. Equivalently, truncate the phylogeny by 
clipping all edges leading to unobserved copy numbers in a postorder traversal.   

\subsubsection*{Partial genomes}
An incomplete genome at a leaf~$\nx$ is 
characterized by the fraction~$(1-\epsilon_{\nx})$ of the genome that is annotated. 
Assuming a simple model of randomly missing copies, we have 
\[
\Probcmd{\tilde{\xi}_{\nx}=k}{\xi_{\nx}=n} = \binom{n}{k} (1-\epsilon_{\nx})^k (\epsilon_{\nx})^{n-k},
\]
where~$\tilde{\xi}$ is the number of annotated copies, and~$\xi$ is the 
true copy number in the complete genome. In other words, the recurrences of 
Theorems~\ref{tm:lik.rec.finite} for the likelihood and~\ref{tm:empty} for the empty profile 
remain the same, with the only change that 
$\tilde{q}_{\nx}\ne q_{\nx}$ at such a leaf with $\epsilon_{\nx}>0$. 
Without constraints, however, the trio~$\bigl(p_{\nx},q_{\nx},\epsilon_{\nx}\bigr)$
is not identifiable: the distribution parameters 
\[
p = p_{\nx}+(1-p_{\nx})\epsilon_{\nx}\frac{1-q_{\nx}}{1-q_{\nx}\epsilon_{\nx}}
\qquad
q = q_{\nx}\frac{1-\epsilon_{\nx}}{1-q_{\nx}\epsilon_{\nx}}
\qquad
\epsilon = 0,
\]
produce the exact same distribution at~$\nx$ as $(p_{\nx},q_{\nx},\epsilon_{\nx})$. 

\subsection{Posterior probabilities for ancestral copy numbers}
Let~$\Xi=\{\xi_{\ny}=n_{\ny}\colon \ny\in\mathcal{L}\}$ be an arbitrary profile
of copy numbers observed at the leaves.  
Theorems~\ref{tm:lik.rec.finite} and~\ref{tm:lik.rec.multi.right} 
provide the recurrences for the conditional 
likelihoods~$\tilde{\nlik}_u$ and~$\tilde{\slik}_u$ 
of the partial profile~$\Xi_u$ conditioned 
on the surviving copies~$\tilde{\xi}_u$ and~$\tilde{\eta}_u$,
respectively. Define the complementary \definehere{outside likelihoods} 
\begin{equation}\label{eq:olik.def}
\nolik_{\nx}(\ell) = \PROB\{\Xi-\Xi_{\nx}, \tilde{\xi}_{\nx}=\ell\}
\quad\text{and}\quad
\solik_{\nx}(s) = \PROB\{\Xi-\Xi_{\nx}, \tilde{\eta}_{\nx}=s\},
\end{equation}
where $\Xi-\Xi_u = \bigl\{\xi_{\ny}=n_{\ny}\colon \ny\in \mathcal{L}-\mathcal{L}_{\nx}\bigr\}$ 
denotes the profile outside the subtree rooted at node~$\nx$. 
\begin{theorem}[Outside likelihoods]\label{tm:olik.rec}
Let~$\Xi=\{\xi_{\ny}=n_{\ny}\colon \ny\in\mathcal{L}\}$ be an arbitrary profile, 
and define the outside likelihoods as in Equation~\eqref{eq:olik.def}. 
The following recurrences hold. 
\begin{enumerate}[label=(\roman*)]
\item At the root, $\solik_{\troot}(0)=1$ and $\solik_{\troot}(s)=0$ for $s>0$. 
\item At any node~$\nx$, for all $0\le \ell \le m_{\nx}$, 
\begin{equation}\label{eq:olik.node}
\nolik_{\nx}(\ell) = \sum_{s=0}^{\ell} 
	\solik_{\nx}(s) \times \binom{\kappa_{\nx}+\ell-1}{\ell-s}(1-\tilde{q}_{\nx})^{\kappa_{\nx}+s}(\tilde{q}_{\nx})^{\ell-s}.
\end{equation}
\item At every non-root node $\ny$ with parent~$\nx$ and sibling $\nz$ (i.e, $\nx\ny,\nx\nz\in T$), 
	for all $0\le s\le m_{\ny}$, 
\begin{multline}\label{eq:olik.edge}
\solik_{\ny}(s) = 
	\sum_{\ell=s}^{m_{\nx}}
		\nolik_{\nx}(\ell)
		\times \tilde{\slik}_{\nz}^{\ell}(\ell-s)
		\\
		\times \binom{\ell}{s}
			\Bigl(\frac{1-\tilde{p}_{\ny}}{1-\tilde{p}_{\ny}\tilde{p}_{\nz}}\Bigr)^{s}
			\Bigl(\frac{\tilde{p}_{\ny}-\tilde{p}_{\ny}\tilde{p}_{\nz}}{1-\tilde{p}_{\ny}\tilde{p}_{\nz}}\Bigr)^{\ell-s}.			
\end{multline}
\end{enumerate}
\end{theorem}

Theorem~\ref{tm:olik.rec} 
with Theorem~\ref{tm:lik.rec.finite}
deliver the posterior probabilities in 
computable formulas. 
\begin{corollary}[Posterior probabilities]\label{cor:posterior}
Fix an arbitrary profile~$\Xi$ and let 
$\tilde{\nlik}_{\nx}, \tilde{\slik}_{\nx}, \nolik_{\nx}, \solik{\nx}$ 
denote the inside and outside likelihoods at every node~$\nx$. 
\begin{enumerate}[label=(\roman*)]
\item The profile likelihood can be computed 
by either formulas 
\begin{equation}\label{eq:lik.node}
L(\Xi) = \PROB\{\Xi\}
	= \sum_{\ell=0}^{m_{\nx}} \nolik_{\nx}(\ell)\times  \tilde{\nlik}_{\nx}(\ell)
 = \sum_{s=0}^{m_{\nx}} \solik_{\nx}(s)\times  \tilde{\slik}_{\nx}(s).
\end{equation}
\item The posterior distribution of~$\tilde{\xi}_{\nx}$ is 
$
\Probcsm{\tilde{\xi}_{\nx}=\ell}{\Xi} 	=\frac{\nolik_{\nx}(\ell)\times \tilde{\nlik}_{\nx}(\ell)}{L(\Xi)}.
$
\item The posterior distribution of $\tilde{\eta}_{\nx}$ is 
$
\Probcsm{\tilde{\eta}_{\nx}=s}{\Xi} = \frac{\solik_{\nx}(s)\times  \tilde{\slik}_{\nx}(s)}{L(\Xi)}
$.
\end{enumerate} 
\end{corollary}

\subsection{Partial derivatives of the likelihood}
Suppose that we are interested in the corrected likelihood 
for a sample of family profiles $\{\Xi_f\colon f=1,\dotsc, F\}$.
By Equation~\eqref{eq:lik.corr},
the derivative of the corrected log-likelihood, 
with respect to any distribution parameter~$\theta$ 
is 
\begin{equation}\label{eq:loglik.d}
\frac{\partial}{\partial\,\theta}\bigl(\ln L^*\bigr)
= \Biggl(\sum_{f=1}^F \frac{L'(\Xi_f)}{L(\Xi_f)} \Biggr)+F \frac{L'(0)}{1-L(0)},
\end{equation}
where $L'(\Xi)=\frac{\partial L(\Xi)}{\partial\, \theta}$ 
denotes the derivative of the uncorrected profile likelihood. 

It is tempting to choose the optimized distribution 
parameters directly as~$\kappa_{\nx}$ and the survival parameters~$\tilde{p}_{\nx}, \tilde{q}_{\nx}$
for the maximization of the corrected log-likelihood~$\ln L^*$. 
They uniquely determine the parameters $p_{\nx}, q_{\nx}$, and, consequently, the edge-specific 
rate parameters.   
The values of $\tilde{p}$ and $\tilde{q}$
are, however, not arbitrary across the tree.  
\begin{theorem}[Unicity of survival parameters]\label{tm:survival.invert}
Let~$T$ be a phylogeny equipped with arbitrary gain rates $0<\kappa_{\nx}$ 
and arbitrary survival parameters $0<\tilde{p}_{\nx}, \tilde{q}_{\nx}<1$
at every node~$\nx$.  
If, at every non-root ancestral node~$\nx$,  
\begin{equation}\tag{*}\label{eq:p.monotone}
\tilde{p}_{\nx}>(1-\tilde{q}_{\nx})\prod_{\nx\ny\in T} \tilde{p}_{\ny},
\end{equation}
then there exists a phylogenetic birth-death model 
on the same phylogeny
with valid distribution parameters $0<p_{\nx}, q_{\nx}<1$ and 
same gain rates $\kappa_{\nx}$. 
Otherwise, no solution exists with positive~$p_{\nx}$ on every edge.   
\end{theorem}

In light of Theorem~\ref{tm:survival.invert}, we should 
aim at using the partial derivatives with respect 
to~$\tilde{p}$ and~$\tilde{q}$ as 
an intermediate step toward 
inferring the dependance on~$p$ and~$q$. 
Using Corollary~\ref{cor:posterior}, we can determine 
the partial derivatives with respect to the 
survival distribution parameters.
\begin{theorem}\label{tm:Ld}
\begin{enumerate}[label=(\roman*)]
\item At every node $1\le \nx\le R$, 
\begin{align*}
\frac{\partial L(\Xi)}{\partial \tilde{q}_{\nx}} & = 
	L(\Xi)\times \biggl(
		\frac1{\tilde{q}_{\nx}}\Expcsm{\tilde{\xi}_{\nx}}{\Xi}
	- \Bigl(\frac1{\tilde{q}_{\nx}}+\frac1{1-\tilde{q}_{\nx}}\Bigr)\Expcsm{\tilde{\eta}_{\nx}}{\Xi}
	-\frac1{1-\tilde{q}_{\nx}}\kappa_{\nx}\biggr).
\end{align*}
\item At every non-root node $1\le \ny < R$,	
\[
\frac{\partial L(\Xi)}{\partial \tilde{p}_{\ny}} 
= L(\Xi)\times \biggl(
	\frac1{\tilde{p}_{\ny}}\Expcsm{\tilde{\xi}_{\nx}}{\Xi}
 - \Bigl(\frac1{\tilde{p}_{\ny}}+\frac{1-\epsilon}{1-\tilde{p}_{\ny}}\Bigr)\Expcsm{\tilde{\eta}_{\ny}}{\Xi}
 \biggr),
\]
where $\epsilon = \frac{\epsilon_{\nx}}{\tilde{p}_{\ny}}=\frac{\prod_{\nx\nz\in T} \tilde{p}_{\nz}}{\tilde{p}_{\ny}}$
is the product of $\tilde{p}_{\nz}$ across the siblings 
with the same parent~$\nx\ny,\nx\nz\in T$.  
\item The partial derivatives with respect to $\kappa_{\nx}$ are, for all $1\le \nx \le R$,
\begin{align*}
\frac{\partial L(\Xi)}{\partial \kappa_{\nx}}
	& = L(\Xi) \times \Biggl(
		\ln(1-\tilde{q}_{\nx})
	 + \sum_{i=0}^{m_{\nx}-1} \frac{\Probcsm{\tilde{\xi}_{\nx}>i}{\Xi}-\Probcsm{\tilde{\eta}_{\nx}>i}{\Xi}}{\kappa_{\nx}+i}
	 \Biggr).
\end{align*}
%
\item The partial derivatives for the empty profile $\Xi=0$ are
\[
\frac{\partial\,L(0)}{\partial\,\tilde{q}_{\nx}}
= -L(0) \frac{\kappa_{\nx}}{1-\tilde{q}_{\nx}},
\quad
\frac{\partial\,L(0)}{\partial\,\tilde{p}_{\nx}}=0
\quad\text{and}\quad
\frac{\partial\,L(0)}{\partial\,\kappa_{\nx}}=
L(0)\times \ln(1-\tilde{q}_{\nx}).
\]
\end{enumerate}
\end{theorem}

Note that using the posterior distributions from Corollary~\ref{cor:posterior}, 
we readily obtain the posterior expectations
\[
\Expcsm{\tilde{\xi}_{\nx}}{\Xi} = \sum_{\ell=0}^{m_{\nx}} \ell  \times \Probcsm{\tilde{\xi}_{\nx}=\ell}{\Xi}%
\quad\text{and}\quad
\Expcsm{\tilde{\eta}_{\nx}}{\Xi} = \sum_{s=0}^{m_{\nx}} s \times \Probcsm{\tilde{\eta}_{\xi}}{\Xi},
\]
as well as the distribution tails $\Probcsm{\tilde{\xi}_{\nx}>i}{\Xi}=\sum_{\ell=i+1}^{m_{\nx}} \Probcsm{\tilde{\xi}_{\nx}=\ell}{\Xi}$
and $\Probcsm{\tilde{\eta}_{\nx}>i}{\Xi}=\sum_{\ell=i+1}^{m_{\nx}} \Probcsm{\tilde{\eta}_{\nx}=\ell}{\Xi}$
which are needed in Theorem~\ref{tm:Ld} and the following Corollary~\ref{cor:loglik.d}, which 
combines Theorem~\ref{tm:Ld} with Equation~\eqref{eq:loglik.d}.
\begin{corollary}\label{cor:loglik.d}
Let $\fdf=\ln L^*$ denote the corrected log-likelihood for a sample of family profiles $\{\Xi_f\colon f=1,\dotsc, F\}$.
Define the posterior expected counts across the sample
\begin{align*}
\tilde{N}_{\nx}& =\sum_{f=1}^F \Expcsm{\tilde{\xi}_{\nx}}{\Xi_f}
& \tilde{S}_{\nx}  & =\sum_{f=1}^F \Expcsm{\tilde{\eta}_{\nx}}{\Xi_f}
\\
\tilde{N}_{\nx}^{>i} & =\sum_{f=1}^F \Probcsm{\tilde{\xi}_{\nx}> i}{\Xi_f}
&
\tilde{S}_{\nx}^{> i}& =\sum_{f=1}^F \Probcsm{\tilde{\eta}_{\nx}> i}{\Xi_f}
\end{align*}
at every node $1\le \nx\le R$. 
\begin{enumerate}[label=(\roman*)] 
\item At every node $1\le \nx\le R$,  
\begin{align*}
\frac{\partial (\ln L^*)}{\partial \tilde{q}_{\nx}}
& = \frac{\tilde{N}_{\nx}-\tilde{S}_{\nx}}{\tilde{q}_{\nx}}-\frac{\tilde{S}_{\nx}+\kappa_{\nx}/\bigl(1-L(0)\bigr)}{1-\tilde{q}_{\nx}}
\end{align*}
\item For a non-root node~$1\le \ny<R$,let $\nx$ be its parent:
\begin{align*}
\frac{\partial (\ln L^*)}{\partial \tilde{p}_{\nx}}
& = \frac{\tilde{N}_{\nx}}{\tilde{p}_{\nx}}-\frac{(1-\epsilon)\tilde{S}_{\ny}}{1-\tilde{p}_{\nx}}
\end{align*}
with $\epsilon = \bigl(\prod_{\nx\nz\in T}\tilde{p}_{\nz}\bigr)/\tilde{p}_{\ny}$. 
\item At every node~$1\le \nx\le R$,
\begin{align*}
\frac{\partial (\ln L^*)}{\partial \kappa_{\nx}}
& = F\frac{\ln(1-\tilde{q}_{\nx})}{1-L(0)}+\sum_{i=0}^{m_{\nx}-1} \frac{\tilde{N}_{\nx}^{>i}-\tilde{S}_{\nx}^{>i}}{\kappa_{\nx}+i}
\end{align*}
\end{enumerate}
\end{corollary}

Powerful numerical algorithms for 
function maximization (conjugate gradient and variable metric  
methods like Broyden-Fletcher-Goldfarb-Shanno)
exploit the gradient for quick convergence to optimum. 
The likelihood optimization for a phylogenetic birth-death model can rely 
on the computation of both the likelihood (Theorem~\ref{tm:lik.rec.finite}), 
and the gradient with respect to the parameters $\kappa_{\nx},p_{\nx},q_{\nx}$ 
across the tree. By Theorem~\ref{tm:param.probs}, the 
probabilistic model is uniquely determined by the parameter set, up to  
equivalent rate scalings.  
Maximizing the likelihood with respect to the survival 
distribution parameters~$\tilde{p}$ and~$\tilde{q}$ from Theorem~\ref{tm:Ld}
is not straightforward because Theorem~\ref{tm:survival.invert} imposes 
monotonicity constraints between parameters on adjoining edges.  
Let~$\nx$ be an arbitrary node at some depth~$d$ (root is at depth~0).  
For a distribution parameter 
such as $\theta_{\ny}=p_{\ny}$ or $\theta_{\ny}=q_{\ny}$,  
\[
\frac{\partial L(\Xi)}{\partial \theta_{\ny}}
	 = \sum_{\nx=1}^R 
		\frac{\partial L(\Xi)}{\partial \tilde{q}_{\nx}} \frac{\partial \tilde{q}_{\nx}}{\partial \theta_{\ny}}
		+ \sum_{\nx=1}^R 
		\frac{\partial L(\Xi)}{\partial \tilde{p}_{\nx}} \frac{\partial \tilde{p}_{\nx}}{\partial \theta_{\ny}},
\]
by the chain rule. 
In particular, $p_{\ny}$ and $q_{\ny}$ influence
$\tilde{p}_{\nx}$ and $\tilde{q}_{\nx}$ at nodes~$\nx$ 
along the path between the root and~$\ny$. Consequently, 
the above sums include only the ancestors of~$\ny$, and 
the partial derivatives can be computed in a preorder traversal.
We state the procedure in a generic theorem about 
recovering the derivatives of 
any function~$\fdf$ of the distribution parameters. 

\begin{theorem}[Gradient computation]\label{tm:fd}
Let~$\fdf$ be an arbitrary differentiable function 
of the distribution parameters $\bigl\{\tilde{p}_{\nx}, \tilde{q}_{\nx}\bigr\}_{\nx=1}^R$. 
Let~$\fdf^{(\theta_{\ny})}=\frac{\partial f}{\partial \theta_{\ny}}$
denote the partial derivative
with respect to any distribution parameter $\theta_{\ny}$. 
The partial derivatives $\fdf^{(p_{\ny})}$ (for non-root $\ny$), 
$\fdf^{(q_{\ny})}$ (for any $\ny$) and $\fdf^{(\epsilon_{\ny})}$
(for non-leaf $\ny$) can be computed in a preorder traversal 
by the following recurrences. 
\begin{enumerate}[label=(\roman*)]
\item At the root $\ny=\troot{}$,
\begin{subequations}\label{eq:fd.root}
\begin{align}
\fdf^{(q_{\troot})} & = \frac{1-\epsilon_{\troot}}{(1-q_{\troot}\epsilon_{\troot} )^2}
	\fdf^{(\tilde{q}_{\troot})}
\\
\fdf^{(\epsilon_{\troot})} & = \frac{1-q_{\troot} }{(1-q_{\troot} \epsilon_{\troot} )^2}
	\Bigl( -q_{\troot} \fdf^{(\tilde{q}_{\troot})}\Bigr).
\end{align}
\end{subequations}
\item At every non-root node $1\le \ny<\troot$,
\begin{subequations}\label{eq:fd.node}
\begin{align}
\fdf^{(p_{\ny})} & = \frac{1-\epsilon_{\ny}}{1-q_{\ny}\epsilon_{\ny}}
	\Bigl(\fdf^{(\tilde{p}_{\ny})}+\epsilon \fdf^{(\epsilon_u)}\Bigr)
\\
\fdf^{(q_{\ny})} & = \frac{1-\epsilon_{\ny}}{(1-q_{\ny}\epsilon_{\ny})^2}
	\Bigl(\fdf^{(\tilde{q}_{\ny})}-(1-p_{\ny}) \epsilon_{\ny}\bigl(\fdf^{(\tilde{p}_{\ny})}+\epsilon \fdf^{(\epsilon_u)}\bigr)\Bigr)
\intertext{and, if $\ny$ is not a leaf,}
\fdf^{(\epsilon_{\ny})} & = \frac{1-q_{\ny}}{(1-q_{\ny}\epsilon_{\ny})^2}
	\Bigl((1-p_{\ny}) \bigl(\fdf^{(\tilde{p}_{\ny})}+\epsilon \fdf^{(\epsilon_u)}\bigr)-q_{\ny} \fdf^{(\tilde{q}_{\ny})}\Bigr)
\end{align}
with the parent $\nx$ and 
\[
\epsilon = \prod_{\nz\colon \nx\nz\in T} \{\nz\ne \ny\} \tilde{p}_{\nz} =\frac{\epsilon_{\nx}}{\tilde{p}_{\ny}}. 
\]
\end{subequations}
\end{enumerate}
\end{theorem}

Theorem~\ref{tm:fd} can be employed with the individual family profiles 
using $\fdf=L(\Xi_f)$ 
and plugging~$L'(\Xi_f)=\fdf^{(\theta)}$ into the 
corrected log-likelihood formula of~\eqref{eq:loglik.d}
for each $f=1,\dotsc, F$ in the sum, as well as for~$L'(0)$. 
But it is more efficient to carry out the procedure 
only once at the end, using~$\fdf=\ln L^*$ directly
with its partial derivatives from Corollary~\ref{cor:loglik.d}. 

For the purposes of likelihood maximization, use a 
parametrization  with the logistic and exponential functions as  
\[
p_{\nx} = \frac{1}{1+e^{-\alpha_{\nx}}}
\quad 
q_{\nx} = \frac{1}{1+e^{-\beta_{\nx}}} 
\quad
\kappa_{\nx} = e^{\gamma_{\nx}}
\]
with unconstrained real-valued parameters 
\[
\alpha_{\nx}=\ln \frac{p_{\nx}}{1-p_{\nx}}
\quad 
\beta_{\nx}=\ln \frac{q_{\nx}}{1-q_{\nx}}
\quad 
\gamma_{\nx}=\ln \kappa_{\nx}.
\]
The partial derivatives of $\fdf=\ln L^*$ (or of a single-profile likelihood $\fdf=L(\Xi)$)
are computed by the chain rule as 
\begin{align*}
\fdf^{(\alpha_{\nx})} & = \fdf^{(p_{\nx})}\frac{\partial p_{\nx}}{\partial \alpha_{\nx}}
	= p_{\nx}(1-p_{\nx}) \fdf^{(p_{\nx})} & \{0<\nx<R\}
\\
\fdf^{(\beta_{\nx})} & = \fdf^{(q_{\nx})}\frac{\partial q_{\nx}}{\partial \beta_{\nx}}
	= q_{\nx}(1-q_{\nx}) \fdf^{(q_{\nx})} & \{0<\nx\le R\}
\\
\fdf^{(\gamma_{\nx})} & = \fdf^{(\kappa_{\nx})}\frac{\partial \kappa_{\nx}}{\partial \gamma_{\nx}}
	= \kappa_{\nx} \fdf^{(\kappa_{\nx})} & \{0<\nx\le R\}
\end{align*}

\subsection{Likelihoods in the no-duplication model}
In the case of~$\lambda_{\ny}=0$ on all edges $\nx\ny\in T$, 
the joint distribution of the random variables multiplies Poisson
and binomial masses: 
\begin{multline*}
\PROB\{\xi_1=n_1,\dotsc,\xi_{\troot}=n_{\troot}, \eta_1=s_1,\dotsc, \eta_{R-1}=s_{R-1}\}
\\
 = \PROB\{\xi_{\troot}=n_{\troot}\}
  \times 
 \prod_{\nx\ny\in T}
 	\biggl( 
 	\binom{n_{\nx}}{s_{\ny}} (1-p_{\ny})^{s_{\ny}} (p_{\ny})^{n_{\nx}-s_{\ny}}
 	\times 
 	e^{-r_{\ny}} \frac{(r_{\ny})^{n_{\ny}-s_{\ny}}}{(n_{\ny}-s_{\ny})!}
 	\biggr).
\end{multline*}
This time we assume a Poisson distribution at the root: $\PROB\{\xi_{\troot} =n\} = e^{-r} r^n/(n!)$, 
and, as before $s_{\troot} =0$ for retrieving the likelihood $L(\Xi)=\slik_{\troot}(0)$.   
The recurrences for the likelihood and the empty profile are adjusted accordingly. 
In particular, 
\begin{equation}\label{eq:slik.rec.poisson}
\slik_{\nx}(s) = \sum_{k=0}^{\infty} e^{-r_{\nx}}\frac{(r_{\nx})^k}{k!} \times \nlik_{\nx}(s+k),
\end{equation}
but the recurrence for $\nlik_{\nx}$ stays the same. 

\begin{theorem}[Empty profile in the no-duplication model]\label{tm:empty.poisson}
Define $\epsilon_\nx$ as in Theorem~\ref{tm:empty}, with 
$\tilde{p}_{\ny} = \bigl(p_{\ny}+(1-p_{\ny})\epsilon_{\ny}\bigr)$
at every non-root node~$\ny$, and
$\tilde{r}_{\nx} = r_{\nx}\bigl(1-\epsilon_{\nx}\bigr)$
at every node~$\nx$.
The 
probability of the empty profile is 
\[
L(0) = \prod_{\nx=1}^R e^{-\tilde{r}_{\nx}}
	= \prod_{\nx=1}^R \exp\Bigl(-r_{\nx}(1-\epsilon_{\nx})\Bigr).
\]
\end{theorem}
The likelihood computations of Theorems~\ref{tm:lik.rec.finite} 
and~\ref{tm:olik.rec} adapt easily to the no-duplication model,
with~$\tilde{\xi}$ and~$\tilde{\eta}$ defined as before. 
Two recurrences change: 
at every node~$\nx$, and for all $0\le s\le m_{\nx}$,
\begin{align*}
\tilde{\slik}_{\nx}(s) 
& = \sum_{\ell= s}^{m_{\nx}}
	\tilde{\nlik}_{\nx}(\ell)	
	\times
	e^{-\tilde{r}_{\nx}} \frac{(\tilde{r}_{\nx})^{\ell-s}}{(\ell-s)!},
\intertext{and, for all $0\le \ell\le m_{\nx}$,}
\nolik_{\nx}(\ell) & = 
\sum_{s=0}^{\ell} \solik_{\nx}(s) \times
	 e^{-\tilde{r}_{\nx}} \frac{(\tilde{r}_{\nx})^{\ell-s}}{(\ell-s)!}.
\end{align*}
Consequently, the derivatives of the profile likelihood are 
\begin{align*}
\frac{\partial L(\Xi)}{\partial \tilde{r}_{\nx}}
& = \frac{\partial}{\partial \tilde{r}_{\nx}} \Bigl(\sum_{\ell=0}^{m_{\nx}} \nolik_{\nx}(\ell)\times \tilde{\nlik}_{\nx}(\ell)\Bigr)
\\
& = \sum_{0\le s\le \ell\le m_{\nx}}
	\solik_{\nx(s)} \times \tilde{\nlik}_{\nx}(\ell)
	\times e^{-\tilde{r}_{\nx}}\frac{(\tilde{r}_{\nx})^{\ell-s}}{(\ell-s)!} 
	\biggl(\frac{\ell-s}{\tilde{r}_{\nx}}-1\biggr)
\\
& = L(\Xi)\times \biggl(
		\frac{\Expcsm{\tilde{\xi}_{\nx}}{\Xi}-\Expcsm{\tilde{\eta}_{\nx}}{\Xi}}{\tilde{r}_{\nx}}-1\biggr).
\end{align*}
at every node $1\le \nx\le \troot$. 
In particular, for the empty profile $\Xi=0$,  
\[
\frac{\partial L(0)}{\partial \tilde{r}_{\nx}}
= \frac{\partial}{\partial \tilde{r}_{\nx}} \Bigl(\prod_{\ny=1}^R e^{-\tilde{r}_{\ny}}\Bigr)
= -L(0).
\]
by Theorem~\ref{tm:empty.poisson}.
Substituting into Equation~\eqref{eq:loglik.d} for the derivatives of 
corrected log-likelihood on a sample of family profiles gives
\[
\frac{\partial (\ln L^*)}{\partial \tilde{r}_{\ny}}
= \frac{\tilde{N}_{\nx}-\tilde{S}_{\nx}}{\tilde{r}_{\nx}}-\frac{F}{1-L(0)}.
\] 
The analogue of Theorem~\ref{tm:fd} is the following claim.
\begin{theorem}[Gradient in the no-duplication model]\label{tm:fd.poisson}
Let~$\fdf$ be an arbitrary differentiable function 
of the distribution parameters $\bigl\{\tilde{p}_{\nx},\tilde{r}_{\nx}\bigr\}_{\nx=1}^R$
in a no-duplication model.  
Let~$\fdf^{(\theta_{\ny})}=\frac{\partial f}{\partial \theta_{\ny}}$
denote the partial derivative
with respect to any distribution parameter~$\theta_{\ny}$. 
The partial derivatives $\fdf^{(p_{\ny})}$ (for non-root $\ny$), 
$\fdf^{(r_{\ny})}$ (for any $\ny$) and $\fdf^{(\epsilon_{\ny})}$
(for non-leaf $\ny$) can be computed in a preorder traversal 
by the following recurrences. 
\begin{enumerate}[label=(\roman*)]
\item At the root $\ny=\troot{}$,
\begin{subequations}\label{eq:fd.root.poisson}
\begin{align}
\fdf^{(r_{\troot})} & = \bigl(1-\epsilon_{\troot}\bigr)
	\fdf^{(\tilde{r}_{\troot})}
\\
\fdf^{(\epsilon_{\troot})} & = -r_{\troot} \fdf^{(\tilde{r}_{\troot})}.
\end{align}
\end{subequations}
\item At every non-root node $1\le \ny<\troot$,
\begin{subequations}\label{eq:fd.node.poisson}
\begin{align}
\fdf^{(p_{\ny})} & = \bigl(1-\epsilon_{\ny}\bigr)
	\Bigl(\fdf^{(\tilde{p}_{\ny})}+\epsilon \fdf^{(\epsilon_u)}\Bigr)
\\
\fdf^{(r_{\ny})} & = \bigl(1-\epsilon_{\ny}\bigr)
	\fdf^{(\tilde{r}_{\ny})}.
\intertext{and, if $\ny$ is not a leaf,}
\fdf^{(\epsilon_{\ny})} & = 
	\bigl(1-p_{\ny}\bigr) 
		\Bigl(\fdf^{(\tilde{p}_{\ny})}+\epsilon \fdf^{(\epsilon_u)}\Bigr)
		-r_{\ny} \fdf^{(\tilde{r}_{\ny})}
\end{align}
with the parent $\nx$ and 
\[
\epsilon = \prod_{\nz\colon \nx\nz\in T} \{\nz\ne \ny\} \tilde{p}_{\nz} =\frac{\epsilon_{\nx}}{\tilde{p}_{\ny}}. 
\]
\end{subequations}
\end{enumerate}
\end{theorem}

Note that the different duplications models can be used in the same tree: some edges can have $\lambda=0$, and some 
$\lambda>0$. In the recurrences for $\tilde{\slik}_{\ny}$ and $\nolik_{\ny}$, 
either the Poisson (if $\lambda_{\ny}=0$) or the P{\'o}lya (if $\lambda_{\ny}>0$)
formulas apply, and the computed derivatives are $\partial r_{\ny}$ or $\partial \kappa_{\ny}$, respectively.

\subsection{Algorithmic complexity}
The set of conditional likelihoods 
$\tilde{\nlik}_{\nx}(\ell)$ and $\tilde{\slik}_{\nx}(s)$ 
for a given profile $\Xi$ 
can be computed in a postorder traversal of the phylogeny 
using Theorem~\ref{tm:lik.rec.finite}.  
The recurrences for $\tilde{\slik}_{\nx}(s)$ from~\eqref{eq:lik.rec.edge}
are straigthforward to implement by embedded loops
over~$0\le s\le\ell\le m_{\nx}$.
Define
\[
\mathsf{h}_{\nx}(s,t) = \binom{\kappa_{\nx}+s+t-1}{t} (1-\tilde{q}_{\nx})^{\kappa_{\nx}+s}(\tilde{q}_{\nx})^t.
\]
\begin{center}\small
\begin{fmpage}{0.8\textwidth} 
\begin{enumerate}[label=\arabic*,nosep]
\item[] // Computing $\tilde{\slik}_{\nx}(s)$ for all $s$ 
\item \textbf{for} $\ell\gets 0,1,\dotsc,m_{\nx}$
\item\quad\ \textbf{for} $s\gets 0,1,\dotsc, \ell$
\item\qquad $\tilde{\slik}_{\nx}(s) 
	\gets \tilde{\slik}_{\nx}(s)+\tilde{\nlik}_{\nx}(\ell)\times \mathsf{h}_{\nx}(s,\ell-s)$
\end{enumerate}
\end{fmpage}
\end{center} 
For the recurrence of~\eqref{eq:lik.rec.edge}, compute $\slik_{\nz}^{s+t}(s)$ 
looping over~$t$ and~$s$ in the opposite direction.
Let 
\[
\mathsf{g}_{\ny}(s,t) = 
	\binom{s+t}{s}
	\biggr(\frac{1-\tilde{p}_{\ny}}{1-\tilde{p}_{\ny}\tilde{p}_{\nz}}\biggr)^s
	\biggr(\frac{\tilde{p}_{\ny}-\tilde{p}_{\ny}\tilde{p}_{\nz}}{1-\tilde{p}_{\ny}\tilde{p}_{\nz}}\biggr)^t.
\]
\begin{center}\small
\begin{fmpage}{0.8\textwidth} 
\begin{enumerate}[label=\arabic*,nosep]
\item[] // Computing $\tilde{\nlik}_{\nx}(\ell)$ for all $\ell$ at~$\nx$ with children $\nx\ny,\nx\nz\in T$
\item \textbf{for} $t\gets m_{\nz}, m_{\nz}-1, \dotsc, 0$ 
\item\quad\ $\slik_{\nz}^t(t) \gets \slik_{\nz}(t)$
\item\quad\ \textbf{for} $s\gets 0,1,\dotsc, m_{\ny}$
\item\qquad\ $\tilde{\nlik}_{\nx}(s+t)\gets \tilde{\nlik}_{\nx}(s+t)+\tilde{\slik}_{\ny}(s)\times \tilde{\slik}_{\nz}^{s+t}(t) 
	\times \mathsf{g}_{\ny}(s,t)$ 
\item\qquad\	$\slik_{\nz}^{(s+1)+t}(t) \gets (1-\tilde{p}_{\nz}) \tilde{\slik}_{\nz}^{(s+1}(t+1)+\tilde{p}_{\nz} \tilde{\slik}_{\nz}^{s+t}(t)$
\end{enumerate}
\end{fmpage}
\end{center}
Note that~$\mathsf{g}_{\nx}(s,t)$ and~$\mathsf{h}_{\nx}(s,t)$ can be computed in constant time.
For instance,  
$\ln \mathsf{h}_{\nx}(s,0) = (\kappa_{\nx}+s)\ln(1-\tilde{q}_{\nx})$, and 
for $t>0$, 
\begin{multline*}
\ln \mathsf{h}_{\nx}(s,t)
	=  (\kappa_{\nx}+s)\ln(1-\tilde{q}_{\nx}) + t \ln\tilde{q}_{\nx}
	\\
	+ \ln` \Gamma(\kappa_{\nx}+s+t)-\ln\Gamma(\kappa_{\nx}+s)-\ln\Gamma(t+1)
\end{multline*}
with the Gamma function $\Gamma(z)=\int_0^{\infty} x^{z-1}e^{-x}\,dx$ (so that $\Gamma(t+1)=t!$).  

The outside likelihoods~$\nolik_{\nx}(\ell)$ and~$\solik_{\nx}(s)$ 
from Theorem~\ref{tm:olik.rec} are computed in a preorder traversal. 
Concomitantly, the posterior 
distributions for~$\tilde{\xi}_{\nx}$ and~$\tilde{\eta}_{\nx}$ 
are obtained by Corollary~\ref{cor:posterior} in the same traversal. 
In addition, during the same preorder traversal, the partial derivatives can be computed 
with respect to all 
$\tilde{p}_{\nx},\tilde{q}_{\nx},\kappa_{\nx}$ 
parameters. 
The running time is quadratic in the total number of observed copies. 
\begin{theorem}[Running time for likelihood computation]\label{tm:time}
Let $\Xi=\{\xi_{\nx}=n_{\nx}\colon \nx\in \mathcal{L}\}$ 
be an arbitrary profile across 
$L=|\mathcal{L}|$ leaves.
The profile likelihood
and all posterior distributions for $\tilde{\xi}_{\nx}$ and $\tilde{\eta}_{\nx}$ 
for all ancestral nodes~$\nx$
can be computed in $O\bigl(hL(L\bar{n}^2+1)\bigr)$ time, where  
$\bar{n}= \frac1L \sum_{\nx=1}^{L} n_{\nx}$
is the average of the copy numbers at the leaves, 
and~$h$ is the phylogeny's height. 
\end{theorem}

In order to get the gradient of the corrected log-likelihood~$\fdf=\ln L^*$
over a sample of~$F$ families, 
first compute 
the partial derivatives 
$\fdf^{(\tilde{p}_{\nx})}$,
$\fdf^{(\tilde{q}_{\nx})}$,
and 
$\fdf^{(\kappa_{\nx})}$
of the corrected log-likelihood 
from the derivatives for the individual profile likelihoods 
using Equation~\eqref{eq:loglik.d}. 
Subsequently, 
the recurrences of Theorem~\ref{tm:fd} compute 
all $\fdf^{(p_{\nx})}$ and 
$\fdf^{(q_{\nx})}$ in a single preorder traversal.

\section{Conclusion}
The mathematical framework 
for phylogenetic gain-loss-duplication models 
provides the clean decomposition of Equation~\eqref{eq:likelihood.history}, 
involving a network of dependent random variables. 
The elementary decomposition 
can be employed with standard Bayesian and likelihood methods,
leading to efficient algorithms 
for a notoriously hard bioinformatics problem. 
A case in point is the fast gradient computation algorithm
reported here.

\section{Calculation}
\subsection{Proof of Theorem~\ref{tm:transition}}
\begin{proof}
First, suppose that $\lambda>0$. 
Decompose $\xi(t)$ as in~\eqref{eq:groups.rv}:
\[
\xi(t) = \zeta_0(t) + \sum_{i=1}^n \zeta_i(t)
\]
where $\zeta_0$ follows P{\'o}lya with parameters~$(\kappa,q)$, 
and $\zeta_i$ are iid shifted geometric with parameters $(p,q)$.  
Now define the random variable $\eta(t)$ 
as the number of conserved copies  
\[
\eta(t) = \sum_{i=1}^n \{\zeta_i(t)>0\}
\] 
Since $\zeta_i$ are independent with $\PROB\{\zeta_i(t)=0\}=p$, 
\begin{equation}\label{eq:survivors.binom}
\Probcmd{\eta(t)=s}{\xi(0)=n} = \binom{n}{s} (1-p)^{s} p^{n-s}.
\end{equation}
Since the $\zeta_i=0$ are immaterial 
in the sum, we 
can condition on $\eta(t)$:  
\begin{multline*}
\Probcmd{\xi(t)=m}{\eta(t)=s,\xi(0)=n} = \Probcmd{\xi(t)=m}{\eta(t)=s}
\\
= \PROB\{\zeta_0 + \zeta_1'+\dotsm + \zeta_s'  = m\},
\end{multline*}
where $\zeta_1'-1$ are iid random variables following
a P{\'o}lya distribution with parameter~1:
\[
\PROB\{\zeta_i'-1 = k\} = (1-q) q^{k} = \binom{1+k-1}{k} (1-q)^1 q^{k}.
\]
Looking specifically at the generator functions: 
\begin{align*}
F_0(z) & = \sum_{i=0}^{\infty} \PROB\{\zeta_0=i\} z^i = \Bigl(\frac{1-q}{1-qz}\Bigr)^{\kappa}
\\
F_i(z) & = \sum_{i=0}^{\infty} \PROB\{\zeta'_i=i\} z^i = \frac{z(1-q)}{1-qz}
\intertext{so} 
F(z) & = \sum_{m=s}^{\infty}\Probcmd{\xi(t)=m}{\eta(t)=s} z^m  
 = F_0(z) \prod_{i=1}^{s} F_i(z) = z^s 	\Bigl(\frac{1-q}{1-qz}\Bigr)^{\kappa+s}.
\end{align*}
Hence, $\xi(t)-\eta(t)$ follows a P{\'o}lya distribution with parameter $(\kappa+\eta(t))$,
and the same tail parameter~$q$.  
Now, 
\begin{multline*}
\Probcmd{\xi(t)=m}{\xi(0)=n}
\\
\begin{aligned}[t]
	& = \sum_{s}\Probcmd{\xi(t)=m}{\eta(t)=s} \Probcmd{\eta(t)=s}{\xi(0)=n} 
	\\ & = \sum_{s} \Probcmd{\xi(t)-\eta(t)=m-s}{\eta(t)=s} \Probcmd{\eta(t)=s}{\xi(0)=n} 
	\\ & = \sum_{s=0}^{\min\{n,m\}}
		\binom{n}{s} (1-p)^s p^{n-s}
		\binom{(\kappa+s) +(m-s)-1}{m-s} (1-q)^{\kappa+s} q^{m-s} ,
\end{aligned}
\end{multline*}
as claimed. 

When $\lambda=0$, 
define~$\zeta_i(t)$ for xenologs and inparalogs: 
\[
\PROB\{\zeta_0(t)=k\} = e^{-r} \frac{r^k}{k!},
\]
and $\zeta_i(t)$ for~$i>0$ are Bernoulli random variables 
with 
\[
\PROB\{\zeta_i(t)=0\} = p \qquad 
\PROB\{\zeta_i(t)=1\} = 1-p. 
\]
We condition on $\eta(t)=\sum_{i=1}^n \{\zeta_i(t)>0\} = \sum_{i=1}^n\zeta_i(t)$
with the same binomial distribution as in~\eqref{eq:survivors.binom}:
now~$\xi(t)-\eta(t)$ has a Poisson distribution. 
\end{proof}

\subsection{Proof Theorem~\ref{tm:empty}}
\begin{proof}
Let $\nlik_{\nx},\slik_{\nx}$ denote the likelihoods 
for the empty profile: 
\begin{align*}
\nlik_{\nx}(n) & = \Probcmd{\forall \ny\in \mathcal{L}_{\nx} \colon \xi_{\ny}=0}{\xi_{\nx}=n}
\\
\slik_{\nx}(s) &= \Probcmd{\forall \ny \in \mathcal{L}_{\nx} \colon \xi_{\ny}=0}{\eta_{\nx}=s}.
\end{align*}
%
Let~$Q_{\nx}$ denote the product of $(1-\tilde{q}_{\ny})^{\kappa_{\ny}}$
across all edges in the subtree of~$\nx$:
$Q_{\nx}=1$ at a leaf, and at an ancestral node~$\nx$
with children $\ny,\nz$
\[
Q_{\nx} = \bigl(Q_{\ny} (1-\tilde{q}_{\ny})^{\kappa_{\ny}}\bigr)
	\bigl(Q_{\nz} (1-\tilde{q}_{\nz})^{\kappa_{\nz}}\bigr).
\] 
We prove that for all nodes~$\nx$,
\[
\slik_{\nx}(s) = Q_{\nx} \times (\epsilon_{\nx})^s (1-\tilde{q}_{\nx})^{\kappa_{\nx}+s} 
\qquad
\nlik_{\nx}(n) = Q_{\nx} \times (\epsilon_{\nx})^n.
\]
(With $0^0=1$ and $0^n=0$ for $n>0$.)
In particular, at the root~$R$, the 
probability of the empty profile is 
\[
L(0) = \slik_{\troot}(0) = Q_{\troot} (1-\tilde{q}_{\troot})^{\kappa_{\troot}}
	= \prod_{\nx=1}^R \biggl(\frac{1-q_{\nx}}{1-q_{\nx}\epsilon_{\nx}}\biggr)^{\kappa_{\nx}}.
\]
We prove the claim 
by induction in the node height, 
starting with the leaves. 

\paragraph{Base case.} At a leaf~$\nx$ (height 0), we have $\nlik_{\nx}(0)=1$
and $\nlik_{\nx}(n)=0$ for $n>0$. 
Since $\epsilon_{\nx}=0$, $\nlik_{\nx}(n)=0^n = \epsilon_{\nx}^n$ holds at all~$n$.  

At any node~$\ny$, with $\kappa=\kappa_{\ny}, q=q_{\ny}, \epsilon=\epsilon_{\ny}, Q = Q_{\ny}$
and $\tilde{q} = q\frac{1-\epsilon}{1-q\epsilon}$: 
\begin{align*}
\slik_{\ny}(s) & = \sum_{n=s}^\infty \binom{\kappa+n-1}{n-s}(1-q)^{\kappa+s}(q)^{n-s}\nlik_{\ny}(n)
\\
		& = Q  \sum_{n=s}^\infty \binom{\kappa+n-1}{n-s}(1-q)^{\kappa+s}(q)^{n-s} \epsilon^n 
\\
		& = Q \epsilon^s \biggl(\frac{1-q}{1-q\epsilon}\biggr)^{\kappa+s}	
\\
		& = Q\epsilon^s (1-\tilde{q})^{\kappa+s}.		
\end{align*}

\paragraph{Induction.} Suppose~$\nx$ is an ancestral node 
with two non-null children~$\ny$ and~$\nz$. The height of~$\nx$ is 
$(h+1)$ for some $h\ge 0$:
suppose that the induction claim holds for all nodes at heights up to~$h$.  
Both children have heights at most~$h$, so 
\[
\slik_{\ny}(s) = Q_{\ny}(1-\tilde{q}_{\ny})^{\kappa_{\ny}} 
	(\epsilon_{\ny}(1-\tilde{q}_{\ny}))^{s}
\qquad
\slik_{\nz}(s) = Q_{\nz}(1-\tilde{q}_{\nz})^{\kappa_{\nz}} (\epsilon_{\nz}(1-\tilde{q}_{\nz}))^{s}.
\]
Therefore, 
\begin{align*}
\nlik_{\nx}(n) & = \biggl(\sum_{s=0}^n \binom{n}{s} (1-p_{\ny})^s (p_{\ny})^{n-s} \slik_{\ny}(s)\biggr)
	\times \biggl(\sum_{s=0}^n \binom{n}{s} (1-p_{\nz})^s (p_{\nz})^{n-s} \slik_{\nz}(s)\biggr)
\\
	& = \bigl(Q_{\ny} (1-\tilde{q}_{\ny})^{\kappa_{\ny}}\bigr)
	\bigl(p_{\ny}+(1-p_{\ny})\epsilon_{\ny}(1-\tilde{q}_{\ny})\bigr)^n 
	\\
	& \quad \times 
	\bigl(Q_{\nz} (1-\tilde{q}_{\nz})^{\kappa_{\nz}}\bigr)
	\bigl(p_{\nz}+(1-p_{\nz})\epsilon_{\nz}(1-\tilde{q}_{\nz})\bigr)^n
\\
	& = Q_{\nx} (\epsilon_{\nx})^n.
\end{align*}
\end{proof} 

\subsection{Proof of Theorem~\ref{tm:lik.rec.finite}}
\begin{proof}
\begin{enumerate}[label=(\roman*)]
\item Given the definition of~$\tilde{\eta}_{\nx}$ and~$\tilde{\xi}_{\nx}$, 
the Pigeonhole Principle implies that their maximal value is 
$m_{\nx}=\sum_{\ny\in \mathcal{L}_{\nx}} n_{\ny}$, the sum of copy numbers at the 
leaves descending from~$\nx$.  
\item 
By Equation~\eqref{eq:cdist.node},
the generating function for the 
conditional distribution 
of $\tilde{\xi}_{\nx} \mid \tilde{\eta}_{\nx}$ is 
\begin{align*}
\tilde{F}_s(z) & = \sum_{\ell=0}^{\infty} \Probcmd{\tilde{\xi}_{\nx}=\ell}{\tilde{\eta}_{\nx}=s}z^\ell
\\
	& = \sum_{n=s}^{\infty} \binom{\kappa_{\nx}+n-1}{n-s}(1-q_{\nx})^{\kappa_{\nx}+s}  (q_{\nx})^{n-s} 
		\sum_{i=0}^{n-s} \binom{n-s}{i}(1-\epsilon_{\nx})^{i}(\epsilon_{\nx})^{n-s-i}z^{s+i}
\\
	& = z^s\sum_{k=0}^{\infty}	\binom{\kappa_{\nx}+s+k-1}{k} (1-q_{\nx})^{\kappa_{\nx}+s} (q_{\nx})^k (\epsilon_{\nx}+(1-\epsilon_{\nx})z)^{k}	
\\
	& = z^{s} \Bigl(\frac{1-\tilde{q}_{\nx}}{1-\tilde{q}_{\nx} z}\Bigr)^{\kappa_{\nx}+s},
\end{align*}
where we used $1-\tilde{q}_{\nx}=\frac{1-q_{\nx}}{1-q_{\nx}\epsilon_{\nx}}$.  
Hence, 
$(\tilde{\xi}_{\nx}-\tilde{\eta}_{\nx})$ 
has a P{\'o}lya distribution with parameters~$(\kappa_{\nx}+\tilde{\eta}_{\nx})$ 
and~$\tilde{q}_{\nx}$:
\begin{equation}\label{eq:survival.node.dist}
\Probcmd{\tilde{\xi}_{\nx}=\ell}{\tilde{\eta}_{\nx}=s}
	= \binom{\kappa_i+\ell-1}{\ell-s} (1-\tilde{q}_{\nx})^{\kappa_{\nx}+s} (\tilde{q}_{\nx})^{\ell-s}.
\end{equation}
Now we have the recurrences for $\tilde{\slik}_{\nx}$:
\begin{align*}
\tilde{\slik}_{\nx}(s) & = \Probcmd{\Xi_{\nx}}{\tilde{\eta}_{\nx}=s}
= \sum_{\ell\ge s}
	\tilde{\nlik}_{\nx}(\ell)	
	\times
	\binom{\kappa_{\nx}+\ell-1}{\ell-s} (1-\tilde{q}_{\nx})^{\kappa_{\nx}+s} (\tilde{q}_{\nx})^{\ell-s},
\end{align*}
as claimed.

\item
The $\tilde{\xi}_{\nx}=\ell$ ancestral copies 
get sorted in the two child lineages with probabilities 
$(1-\tilde{p}_{\ny})\tilde{p}_{\nz}/(1-\tilde{p}_{\ny}\tilde{p}_{\nz})$, 
$(1-\tilde{p}_{\nz})\tilde{p}_{\ny}/(1-\tilde{p}_{\ny}\tilde{p}_{\nz})$, and 
$(1-\tilde{p}_{\ny})(1-\tilde{p}_{\nz})/(1-\tilde{p}_{\ny}\tilde{p}_{\nz})$ 
as conserved only on the left~$\ny$, only on the right~$\nz$, 
or on both sides.  
Hence, 
\begin{align}
\Probcmd{\tilde{\eta}_{\ny}=s}{\tilde{\xi}_{\nx}=\ell}
	& = \binom{\ell}{s} \Bigl(\frac{1-\tilde{p}_{\ny}}{1-\tilde{p}_{\ny}\tilde{p}_{\nz}}\Bigr)^s \Bigl(\tilde{p}_{\ny} \frac{1-\tilde{p}_{\nz}}{1-\tilde{p}_{\ny}\tilde{p}_{\nz}}\Bigr)^{\ell-s}
	\label{eq:cdist.u}
	\\
\Probcmd{\tilde{\eta}_{\ny}=s}{\tilde{\xi}_{\nx}=\ell}
	& = \binom{\ell}{s} \Bigl(\frac{1-\tilde{p}_{\nz}}{1-\tilde{p}_{\ny}\tilde{p}_{\nz}}\Bigr)^s \Bigl(\tilde{p}_{\nz}\frac{1-\tilde{p}_{\ny}}{1-\tilde{p}_{\ny}\tilde{p}_{\nz}}\Bigr)^{\ell-s}
	\notag
\end{align}
for $0\le s\le \ell$. 
Define~$\tilde{\psi}_{\nx}$ as the ancestral copies from~$\tilde{\xi}_{\nx}$
that survive in both child lineages:
\begin{align}
\Probcmd{\tilde{\psi}_{\nx}=k}{\tilde{\eta}_{\ny}=s}
	= & \binom{s}{k} (1-\tilde{p}_{\nz})^k (\tilde{p}_{\nz})^{s-k}
	\label{eq:cdist.d} 
	\\
\Probcmd{\tilde{\psi}_{\nx}=k}{\tilde{\eta}_{\nz}=s}
	= & \binom{s}{k} (1-\tilde{p}_{\ny})^k (\tilde{p}_{\ny})^{s-k}
	\notag
\end{align}
for $0\le k\le s$.  The two random 
variables~$\tilde{\eta}_{\ny},\tilde{\eta}_{\nz}$ 
are not independent 
when conditioned on~$\tilde{\xi}_{\nx}$, since 
$\tilde{\eta}_{\nz} = \tilde{\xi}_{\nx}-\tilde{\eta}_{\ny}+\tilde{\psi}_{\nx}$:  
\begin{align*}
\Probcmd{\tilde{\eta}_{\ny}=s,\tilde{\eta}_{\nz}=t}{\tilde{\xi}_{\nx}=\ell}
	& = \Probcmd{\tilde{\eta}_{\ny}=s}{\tilde{\xi}_{\nx}=\ell}
		\Probcmd{\tilde{\psi}_{\nx}=(s+t)-\ell}{\tilde{\eta}_{\ny}=s}.
\end{align*}
Combining~\eqref{eq:cdist.u} and~\eqref{eq:cdist.d}
gives us the recurrence for $\tilde{\nlik}_{\nx}$: 
\begin{align*}
\tilde{\nlik}_{\nx}(\ell)
 & = \Probcmd{\Xi_{\nx}}{\tilde{\xi}_{\nx}=\ell}
 \\
 & = \sum_{s+t\ge\ell}^{s,t\le \ell} \Probcmd{\Xi_{\ny}}{\tilde{\eta}_{\ny}=s}\Probcmd{\Xi_{\nz}}{\tilde{\eta}_{\nz}=t}
 	\Probcmd{\tilde{\eta}_{\ny}=s,\tilde{\eta}_{\nz}=t}{\tilde{\xi}_{\nx}=\ell}
 \\
 & = \sum_{s=0}^{\ell} 
 		\biggl(
 			\begin{aligned}[t]
 		& \tilde{\slik}_{\ny}(s) \times \binom{\ell}{s} \Bigl(\frac{1-\tilde{p}_{\ny}}{1-\tilde{p}_{\ny}\tilde{p}_{\nz}}\Bigr)^s \Bigl(\tilde{p}_{\ny} \frac{1-\tilde{p}_{\nz}}{1-\tilde{p}_{\ny}\tilde{p}_{\nz}}\Bigr)^{\ell-s} 		
 		\\
 		\times & \sum_{k=0}^{s} \tilde{\slik}_{\nz}(\ell-s+k)\times \binom{s}{k} (1-\tilde{p}_{\nz})^k (\tilde{p}_{\nz})^{s-k} \biggr).
 		\end{aligned}
\end{align*}
The inner sum can be computed in~$O(1)$ amortized time by dynamic programming: 
for all $0\le d\le\ell$, let 
\[
\tilde{\slik}_{\nz}^{\ell}(d)=\sum_{k=0}^{\ell-d} \tilde{\slik}_{\nz}(d+k)\times \binom{\ell-d}{k} (1-\tilde{p})^k \tilde{p}^{\ell-d-k}, 
\]
with $\tilde{p}=\tilde{p}_{\nz}$, so that 
\begin{align*}
\tilde{\nlik}_{\nx}(\ell)
 & = \sum_{s=0}^{\ell} 
 		\tilde{\slik}_{\ny}(s) 
 		\times \tilde{\slik}^{\ell}_{\nz}(\ell-s)
 		\times \binom{\ell}{s} \Bigl(\frac{1-\tilde{p}_{\ny}}{1-\tilde{p}_{\ny}\tilde{p}_{\nz}}\Bigr)^s \Bigl(\frac{\tilde{p}_{\ny} -\tilde{p}_{\ny} \tilde{p}_{\nz}}{1-\tilde{p}_{\ny}\tilde{p}_{\nz}}\Bigr)^{\ell-s} 		
\end{align*}
The initial values are 
\[
\tilde{\slik}_{\nz}^{\ell}(\ell)=\tilde{\slik}_{\nz}(\ell). 
\]
Let~$s=\ell-d$. Since 
\begin{align*}
\binom{s}{k}(1-\tilde{p})^k\tilde{p}^{s-k}
 = & \{k<s\} \tilde{p} \binom{s-1}{k}(1-\tilde{p})^k \tilde{p}^{(s-1)-k} 
 \\
	+ & \{0<k\} (1-\tilde{p})\binom{s-1}{k-1}(1-\tilde{p})^{k-1} \tilde{p}^{(s-1)-(k-1)}, 
\end{align*}
we have the recursions for $d<\ell$: 
\begin{align*}
\tilde{\slik}^{\ell}_{\nz}(d)& =
	\sum_{k=0}^{\ell-d} \tilde{\slik}_{\nz}(d+k)\times \binom{\ell-d}{k} (1-\tilde{p})^k \tilde{p}^{\ell-d-k}
\\
	& = \tilde{p} \sum_{k=0}^{\ell-d-1}
		\tilde{\slik}_{\nz}(d+k) \binom{\ell-d-1}{k}(1-\tilde{p})^k \tilde{p}^{\ell-d-1-k}\\
	& + (1-\tilde{p}) \sum_{k=1}^{\ell-d} 
		\tilde{\slik}_{\nz}(d+k) 
		\binom{\ell-d-1}{k-1}(1-\tilde{p})^{k-1} \tilde{p}^{(\ell-d-1)-(k-1)}. 
\intertext{By 
setting $d+k = (d+1)+(k-1)$ in the second term,}
\tilde{\slik}^{\ell}_{\nz}(d)	& = \tilde{p} \tilde{\slik}^{\ell-1}_{\nz}(d)
	+(1-\tilde{p}) \tilde{\slik}^{\ell}_{\nz}(d+1).
\end{align*}
\end{enumerate}
\end{proof}

\subsection{Proof of Theorem~\ref{tm:lik.rec.multi.right}}
Consider first a resolution into a left-leaning binary tree, and 
survival in 1, 2, \ldots, $d$ child subtrees 
incrementally. The extinction probabilities for general arity are
\[ 
\epsilon_{\nx,i}
= \prod_{j=1}^i \underbrace{\bigl(p_{\ny_j}+(1-p_{\ny_j})\epsilon_{\ny_j}(1-\tilde{q}_{\ny_j})\bigr)}_{\CB{=\tilde{p}_{\ny_j}}}
\quad\text{for all $1\le i\le d$, and}\quad 
\epsilon_{\nx}=\epsilon_{\nx,d}
\]
at an ancestral node~$\nx$ with children $\nx\ny_1,\dotsc, \nx\ny_d\in T$
in an arbitrary order. 

\begin{theorem}[Likelihood recurrence for multifurcating node]\label{tm:lik.rec.multi}
Let~$\nx$ be a node in a degenerate phylogeny with $d\ge 2$ distinct children 
$\nx\ny_1,\dotsc, \nx\ny_d\in T$ indexed in any order. 
Define the likelihoods $\tilde{\nlik}_{\nx}^i(\ell)$ 
conditioned on~$\ell$ surviving copies 
in the subtrees of $\ny_1,\dotsc, \ny_i$:
\begin{subequations}\label{eq:lik.rec.multi}
\begin{align}
\tilde{\nlik}_{\nx}^1(\ell) & = \slik_{\ny_1}(\ell) \qquad \{ 0\le \ell\le m_{\ny_1}\}
\intertext{and, for all $0\le i< d$ and for all $0\le \ell \le m_{\ny_1}+\dotsb+m_{\ny_{i+1}}$}
\tilde{\nlik}_{\nx}^{i+1}(\ell) & = 
	\begin{aligned}[t]
	\sum_{s=0}^{\min\{\ell, m_{\ny_1}+\dotsb+m_{\ny_{i}}\}} 
		& 
		\tilde{\nlik}_{\nx}^{i}(s) 
			\times \tilde{\slik}_{\ny_{i+1}}^{\ell}(\ell-s)
		\\	
		\times & \binom{\ell}{s} \biggl(\frac{1-\epsilon_{\nx,i}}{1-\epsilon_{\nx,i+1}}\biggr)^s
		\biggl(\frac{\epsilon_{\nx,i}-\epsilon_{\nx,i+1}}{1-\epsilon_{\nx,i+1}}\biggr)^{\ell-s}.
		\end{aligned}
\end{align}
\end{subequations}
Then 
\[
\tilde{\nlik}_{\nx}(\ell) = \tilde{\nlik}_{\nx}^d(\ell).
\]
\end{theorem}
\begin{proof}
In order to accommodate a multifurcation at~$\nx$,       
imagine a resolution of~$\nx$ into $(d-1)$ binary nodes $\{\nx_2,\dotsc, \nx_{d}\}$ 
with $\nx_{d}=\nx$, set $\nx_1=\ny_1$, and define the edges 
\[
T' = T-\{\nx\ny\colon \nx\ny\in T\}
\cup 
\bigl\{\nx_i\nx_{i-1}, \nx_i\ny_i\}_{i=2}^{d}.
\]  
The corresponding random variables are  
$\tilde{\xi}_{\nx_1},\dotsc, \tilde{\xi}_{\nx_{d}}$, so that 
each $\tilde{\xi}_i$ denotes survival in lineages $\nx\ny_1,\dotsc, \nx\ny_{i}$.  
Edges $\nx_{i}\nx_{i-1}$ have length~0, so that 
their distribution parameters are~$p=0$ and~$q=0$. 
Applying the recurrences of Theorem~\ref{tm:lik.rec.finite} 
to the resolved nodes~$\tilde{\nx_i}$ in~$T'$ 
give the recurrences for the multifurcating node
(with $\tilde{\nlik}_{\nx_i}=\tilde{\nlik}_{\nx}^i$ in Equations~\ref{eq:lik.rec.multi}).  
\end{proof}

Theorem~\ref{tm:lik.rec.multi} is based on resolving a multifurcation into a left-leaning 
binary tree. Alternatively, the node can be resolved into a right-leaning binary tree with 
$(d-2)$ right edges of length~0, giving Theorem~\ref{tm:lik.rec.multi.right}.

\begin{proof}[Proof of Theorem~\ref{tm:lik.rec.multi.right}.]
The Theorem combines the techniques of Theorems~\ref{tm:lik.rec.finite} and~\ref{tm:lik.rec.multi}. 
For Equation~\eqref{eq:lik.rec.multi.node}, note 
that $\epsilon_{\nx,-(i-1)} = \tilde{p}_{\ny_{i-1}}\epsilon_{\nx,-i}$ for all $i<d$. 
\end{proof}

\subsection{Proof of Theorem~\ref{tm:olik.rec}}
\begin{proof}
At the root, $\solik_{\troot}(s)=\PROB\{\tilde{\eta}_{\troot}=s\}=\{s=0\}$ by our model. 
Let~$\nx$ be an arbitrary node, and let~$\kappa=\kappa_{\nx}, \tilde{q}=\tilde{q}_{\nx}$. 
Using Equation~\eqref{eq:survival.node.dist}
for $\tilde{\xi}_{\nx}|\tilde{\eta}_{\nx}$, 
\begin{align*}
\nolik_{\nx}(\ell) & = \PROB\{\Xi-\Xi_{\nx}, \tilde{\xi}_{\nx}=\ell\}
\\
	& = \sum_{s} \Probcmd{\Xi-\Xi_{\nx}, \tilde{\xi}_{\nx}=\ell}{\tilde{\eta}_{\nx}=s}
		\PROB\{\tilde{\eta}_{\nx}=s\}
\\
	& = \sum_s \PROB\{\Xi-\Xi_{\nx}, \tilde{\eta}_{\nx}=s\}\Probcmd{\tilde{\xi}_{\nx}=\ell}{\tilde{\eta}_{\nx}=s}
\\
	& = \sum_s \solik_{\nx}(s) \binom{\kappa+\ell-1}{\ell-s}(1-\tilde{q})^{\kappa+s} (\tilde{q})^{\ell-s}.
\end{align*}

Now let~$\nx\ny\in T$ be a parent-child pair, and 
let~$\nx\nz\in T$ be the sibling lineage (with $\ny\ne \nz$).
Since $\Xi-\Xi_{\ny} = (\Xi-\Xi_{\nx})\cup \Xi_{\nz}$, 
\begin{align*}
\solik_{\ny}(s) & = \PROB\{\Xi-\Xi_{\ny}, \tilde{\eta}_{\ny}=s\}
\\
	& = \sum_{\ell} \Probcmd{\Xi-\Xi_{\ny}, \tilde{\eta}_{\ny}=s}{\tilde{\xi}_{\nx}=\ell}
		\PROB\{\tilde{\xi}_{\nx}=\ell\}
\\
	& = \sum_{\ell} 
		\begin{aligned}[t]
		& \PROB\bigl\{\Xi-\Xi_{\nx}, \tilde{\xi}_{\nx}=\ell\bigr\}
		\times \Probcmd{\tilde{\eta}_{\ny}=s}{\tilde{\xi}_{\nx}=\ell}
	\\
	\times &  
		\sum_t \Probcmd{\tilde{\eta}_{\nz}=t}{\tilde{\xi}_{\nx}=\ell, \tilde{\eta}_{\ny}=s}
			\Probcmd{\Xi_{\nz}}{\tilde{\eta}_{\nz}=t}
		\end{aligned}
\\
	& = \sum_{\ell} \nolik_{\nx}(\ell) 
		\times 
		\binom{\ell}{s}\Bigl(\frac{1-\tilde{p}_{\ny}}{1-\tilde{p}_{\ny}\tilde{p}_{\nz}}\Bigr)^s
		\Bigl(\frac{\tilde{p}_{\ny}-\tilde{p}_{\ny}\tilde{p}_{\nz}}{1-\tilde{p}_{\ny}\tilde{p}_{\nz}}\Bigr)^{\ell-s}
		\times \tilde{\slik}_{\nz}^{\ell}(\ell-s),
\end{align*}   
where we used Equations~\eqref{eq:cdist.u} and~\eqref{eq:cdist.d} as in the proof of Theorem~\ref{tm:lik.rec.finite}. 
\end{proof}

\subsection{Proof of Theorem~\ref{tm:survival.invert}}
\begin{proof}
First, define $\epsilon_{\nx}$ at all nodes using $\tilde{p}$: if $\nx$ is a leaf, then 
$\epsilon_{\nx}=0$, and at an ancestral node $\nx$,  
$\epsilon_{\nx} = \prod_{\nx\ny\in T} \tilde{p}_{\ny}$. 
Since all $\tilde{p}_{\nx}$ are positive, $\epsilon_{\nx}>0$ at every ancestral node~$\nx$.  

Let~$\nx$ be an arbitrary node and let $\tilde{q}=\tilde{q}_{\nx}, \tilde{p}=\tilde{p}_{\nx}$. 
Since $0<(1-\tilde{q})\epsilon_{\nx}<1$,  
the equation
$
\tilde{q} = q_{\nx} \frac{1-\epsilon_{\nx}}{1-q_{\nx}\epsilon_{\nx}}
$
has a unique positive solution 
\[
q_{\nx} = \frac{\tilde{q}}{1-(1-\tilde{q})\epsilon_{\nx}}
= \frac{\tilde{q}}{\tilde{q}+(1-\tilde{q})(1-\epsilon_{\nx})}<1.
\]
Furthermore, the equation 
$
\tilde{p}=p_{\nx}+(1-p_{\nx})\epsilon_{\nx}(1-\tilde{q})
$
has a unique solution 
\[
p_{\nx} =  \frac{\tilde{p}-\epsilon_{\nx}(1-\tilde{q})}{1-\epsilon_{\nx}(1-\tilde{q})}<1.
\] 
Since $\epsilon_{\nx}=\prod_{\nx\ny\in T}\tilde{p}_{\ny}$, 
by the assumption of~\eqref{eq:p.monotone}, $p_{\nx}>0$. 
If the assumption is violated by 
$\tilde{p}<\epsilon_{\nx}(1-\tilde{q})$, then $p_{\nx}<0$, which is illegal.  

Since $0<p_{\nx}<1$ and $0<q_{\nx}<1$ can be selected at every node, 
Theorem~\ref{tm:param.probs} implies that a
corresponding phylogenetic birth-death model exists that is unique 
up to equivalent rate scalings.      
\end{proof}

\subsection{Proof of Theorem~\ref{tm:Ld}}
\begin{proof}\small
By Equations~\eqref{eq:lik.node} and~\eqref{eq:olik.node},
\begin{align*}
\frac{\partial L(\Xi)}{\partial \tilde{q}_{\nx}}
& = \sum_{\ell=0}^{m_{\nx}} \tilde{\nlik}_{\nx}(\ell)\frac{\partial\, \nolik_{\nx}(\ell)}{\partial\, \tilde{q}_{\nx}}
 \\
 & =\sum_{\ell=0}^{m_{\nx}} \tilde{\nlik}_{\nx}(\ell) 
 	\frac{\partial}{\partial\,\tilde{q}_{\nx}}
 	\biggl(\sum_{s=0}^{\ell} \solik_{\nx}(s) \times \binom{\kappa_{\nx}+\ell-1}{\ell-s} (1-\tilde{q}_{\nx})^{\kappa_{\nx}+s} (\tilde{q}_{\nx})^{\ell-s}\biggr).
\end{align*}
So, 
\begin{multline}
\frac{\partial L(\Xi)}{\partial \tilde{q}_{\nx}}
= \sum_{0\le s\le \ell\le m_{\nx}} 
 \tilde{\nlik}_{\nx}(\ell) \times \solik_{\nx}(s) 
 \\ 
 	\times \binom{\kappa_{\nx}+\ell-1}{\ell-s} (1-\tilde{q}_{\nx})^{\kappa_{\nx}+s}(\tilde{q}_{\nx})^{\ell-s}
 	\biggl(\frac{\ell-s}{\tilde{q}_{\nx}}-\frac{\kappa_{\nx}+s}{1-\tilde{q}_{\nx}}\biggr).
\end{multline}
By Theorem~\ref{tm:empty}, the empty profile likelihood is 
$L(0)=\prod_{\nx=1}^R (1-\tilde{q}_{\nx})^{\kappa_{\nx}}$, so 
\[
\frac{\partial\,L(0)}{\partial\,\tilde{q}_{\nx}}
= -L(0) \frac{\kappa_{\nx}}{1-\tilde{q}_{\nx}}.
\]

For derivatives with respect to~$\tilde{p}_{\ny}$ on an edge between 
a non-root node~$\ny$ and its parent
$\nx\ny\in T$, 
consider the recurrences of Theorems~\ref{tm:lik.rec.finite} and~\ref{tm:lik.rec.multi.right}.
Both can be written as
\[
\tilde{\nlik}_{\nx}(\ell)
	= \sum_{s=0}^{\ell}
		\tilde{\slik}_{\ny}(s)
		\times \tilde{\slik}_{-\ny}^{\ell}(\ell-s)
		\times \binom{\ell}{s} 
			\biggl(\frac{1-\tilde{p}_{\ny}}{1-\tilde{p}_{\ny}\epsilon}\biggr)^s
			\biggl(\frac{\tilde{p}_{\ny}-\tilde{p}_{\ny}\epsilon}{1-\tilde{p}_{\ny}\epsilon}\biggr)^{\ell-s}.
\]
At a binary node~$\nx$ (Theorem~\ref{tm:lik.rec.finite}), 
$\epsilon = \tilde{p}_{\nz}$ with the sibling~$\nx\nz\in T$, and 
$\tilde{\slik}_{-\ny}^{\ell}(k)=\tilde{\slik}_{\nz}^{\ell}(k)$.
If $\nx$ has more than 2 children $\ny_1,\dotsc, \ny_d$, then order them so that~$\ny$ is the first, 
and apply Theorem~\ref{tm:lik.rec.multi.right}: 
$\epsilon = \epsilon_{\nx,-2} = \prod_{j=2}^d \tilde{p}_{\ny_j}$ and 
$\tilde{\slik}_{-\ny}^{\ell}(k)=\tilde{\slik}_{\ny_{2..d}}^{\ell}(k)$
from Equation~\eqref{eq:lik.rec.multi.node}. 
Hence, using Corollary~\ref{cor:posterior}, 
\begin{align*}
\frac{\partial L(\Xi) }{\partial \tilde{p}_{\ny}} 
	& = \sum_{\ell=0}^{m_{\nx}} \nolik_{\nx}(\ell)\frac{\partial \tilde{\nlik}_{\nx}(\ell)}{\partial \tilde{p}_{\ny}}
\\
	& = 
	\sum_{\ell=0}^{m_{\nx}}
	\nolik_{\nx}(\ell) 
\sum_{s=0}^{\min\{\ell,m_{\nx}\}}
		\tilde{\slik}_{\ny}(s) 
		\times \tilde{\slik}_{-\ny}^{\ell}(\ell-s)
		\\ & \times
		\frac{\partial}{\partial \tilde{p}_{\ny}}
	\Biggl(	\binom{\ell}{s} 
			\biggl(\frac{1-\tilde{p}_{\ny}}{1-\tilde{p}_{\ny}\epsilon}\biggr)^s
			\biggl(\frac{\tilde{p}_{\ny}-\tilde{p}_{\ny}\epsilon}{1-\tilde{p}_{\ny}\epsilon}\biggr)^{\ell-s}
			\Biggr).
\end{align*}
Therefore, 
\begin{multline}
\frac{\partial L(\Xi) }{\partial \tilde{p}_{\ny}} 
= \sum_{0\le s\le \ell\le m_{\nx}}
	\nolik_{\nx}(\ell) 
		\times
			\tilde{\slik}_{\ny}(s) 
		\times \tilde{\slik}_{-\ny}^{\ell}(\ell-s)
\\		
		\times \binom{\ell}{s} 
		\biggl(\frac{1-\tilde{p}_{\ny}}{1-\tilde{p}_{\ny}\epsilon}\biggr)^s
			\biggl(\frac{\tilde{p}_{\ny}-\tilde{p}_{\ny}\epsilon}{1-\tilde{p}_{\ny}\epsilon}\biggr)^{\ell-s}
			 \biggl(\frac{\ell-s}{\tilde{p}_{\ny}}-\frac{s(1-\epsilon)}{1-\tilde{p}_{\ny}}\biggr).
\end{multline}
The derivatives for the empty profile likelihood are trivial, since $L(0)$ does not depend on
any of the~$\tilde{p}_{\ny}$.

By Corollary~\ref{cor:posterior} and Theorem~\ref{tm:olik.rec},
\begin{align*}
\frac{\partial\, L(\Xi)}{\partial\,\kappa_{\nx}} 
& = \sum_{\ell=0}^{m_{\nx}} \tilde{\nlik}{\nx}(\ell) \sum_{s=0}^{\ell}
	\solik_{\nx}(s) 
	\frac{\partial}{\partial \kappa_{\nx}}
		\biggl((1-\tilde{q}_{\nx})^{\kappa_{\nx}+s}(\tilde{q}_{\nx})^{\ell-s}
			\binom{\kappa_{\nx}+\ell-1}{\ell-s}\biggr).
\end{align*}
Since
\[
\frac{\partial}{\partial \kappa}
\biggl(\ln \binom{\kappa+\ell-1}{\ell-s}\biggr)
= \frac{\frac{\partial \binom{\kappa+\ell-1}{\ell-s}}{\partial \kappa}}{\binom{\kappa+\ell-1}{\ell-s}},
\]
and 
\[
\frac{\partial}{\partial \kappa}
\biggl(\ln \binom{\kappa+\ell-1}{\ell-s}\biggr)
= \sum_{i=0}^{\ell-s-1} \frac{\partial \ln (\kappa+s+i)}{\partial \kappa}
= \sum_{i=s}^{\ell-1} \frac1{\kappa+i},
\]
we have 
\begin{align*}
\frac{\partial\, L(\Xi)}{\partial\,\kappa_{\nx}} 
& = \ln(1-\tilde{q}_{\nx}) \times L(\Xi)
\\ &
	+ \sum_{0\le s\le \ell\le m_{\nx}}
		\begin{aligned}[t]
		&
		\tilde{\nlik}_{\nx}(\ell) \times \solik_{\nx}(s) 
		\\
	& \times
	\binom{\kappa_{\nx}+\ell-1}{\ell-s}(1-\tilde{q}_{\nx})^{\kappa_{\nx}+s}(\tilde{q}_{\nx})^{\ell-s}
	\Biggl(\sum_{i=s}^{\ell-1}\frac1{\kappa_{\nx}+i}\Biggr).
	\end{aligned}
\\
  & = \ln(1-\tilde{q}_{\nx}) \times L(\Xi) 
  	+\sum_{i=0}^{m_{\nx}-1}\frac1{\kappa_{\nx}+i}\sum_{s=0}^{i}\sum_{\ell=i+1}^{m_{\nx}}
  	\tilde{\nlik}_{\nx}(\ell) \times \solik_{\nx}(s) \times \binom{\kappa_{\nx}+\ell-1}{\ell-s}(1-\tilde{q}_{\nx})^{\kappa_{\nx}+s}(\tilde{q}_{\nx})^{\ell-s}
\\
	& = \ln(1-\tilde{q}_{\nx}) \times L(\Xi) 
  	+\sum_{i=0}^{m_{\nx}-1}
  		\frac1{\kappa_{\nx}+i} \Bigl(\PROB\bigl\{\tilde{\xi}_{\nx}> i; \Xi\bigr\}-\PROB\bigl\{\tilde{\eta}_{\nx}> i; \Xi\bigr\}\Bigr).
\end{align*}

For the empty profile, 
\begin{align*}
\frac{\partial\, L(0)}{\partial\,\kappa_{\nx}} & = 
\frac{\partial}{\partial\,\kappa_{\nx}} \Biggl(\prod_{\ny=1}^R (1-\tilde{q}_{\ny})^{\kappa_{\ny}}\Biggr)
= L(0)\times \ln(1-\tilde{q}_{\nx})
\end{align*}
%
\end{proof}

\subsection{Proof of Theorem~\ref{tm:fd}}
\begin{proof}
Let $\nx_0 \nx_{1}\dotsm \nx_{d-1}$ denote the path between~$\nx_d=\ny$ and the root $\nx_0=\troot{}$
with edges $\nx_{i}\nx_{i+1}\in T$.   
Since~$p_{\ny}$ and~$q_{\ny}$ influence 
$\tilde{p}_{\nx}$ and $\tilde{q}_{\nx}$ at $\nx=\ny$ and at all 
the other ancestors $\nx=\nx_i$, but not at any other node, 
\[
\fdf^{(\theta_{\ny})}
	 = \sum_{i=0}^d
	 	\biggl( 
		\fdf^{(\tilde{q}_{\nx_i})} \frac{\partial \tilde{q}_{\nx_i}}{\partial \theta_{\ny}}
		+ 
		\fdf^{(\tilde{p}_{\nx_i})} \frac{\partial \tilde{p}_{\nx_i}}{\partial \theta_{\ny}}
		\biggr).
\]
Recall the definitions $\tilde{q}_{\ny}=q_{\ny}\frac{1-\epsilon_{\ny}}{1-q_{\ny}\epsilon_{\ny}}$ 
and $\tilde{p}_{\ny} = \frac{p_{\ny}(1-\epsilon_{\ny})+\epsilon_{\ny}(1-q_{\ny})}{1-q_{\ny}\epsilon_{\ny}}$
(substituting $p_{\troot} =0$ at the root).
We have thus
\begin{align*}
\frac{\partial \tilde{p}_{\ny}}{\partial p_{\ny}} & = \frac{1-\epsilon_{\ny}}{1-q_{\ny}\epsilon_{\ny}}
&
\frac{\partial \tilde{q}_{\ny}}{\partial p_{\ny}} & = 0
\\
\frac{\partial \tilde{p}_{\ny}}{\partial q_{\ny}} & = \frac{-(1-p_{\ny})\epsilon_{\ny}(1-\epsilon_{\ny})}{(1-q_{\ny}\epsilon_{\ny})^2}
&
\frac{\partial \tilde{q}_{\ny}}{\partial q_{\ny}} & = \frac{1-\epsilon_{\ny}}{(1-q_{\ny}\epsilon_{\ny})^2}
\\
\frac{\partial \tilde{p}_{\ny}}{\partial \epsilon_{\ny}} & = \frac{(1-p_{\ny})(1-q_{\ny})}{(1-q_{\ny}\epsilon_{\ny})^2}
& 
\frac{\partial \tilde{q}_{\ny}}{\partial \epsilon_{\ny}} & = \frac{-q_{\ny}(1-q_{\ny})}{(1-q_{\ny}\epsilon_{\ny})^2}.
\end{align*}
\begin{enumerate}[label=(\roman*)]
\item 
If $\ny=\troot$ is the root, then
\begin{align*}
\fdf^{(q_{\troot})}
& = \fdf^{(\tilde{p}_{\troot})}\frac{\partial \tilde{p}_{\troot}}{\partial q_{\troot}}
	+ \fdf^{(\tilde{q}_{\troot})}\frac{\partial \tilde{q}_{\troot}}{\partial q_{\troot}}
= -\fdf^{(\tilde{p}_{\troot})} \frac{\epsilon_{\troot}(1-\epsilon_{\troot})}{(1-q_{\troot}\epsilon_{\troot})^2}
	+ \fdf^{(\tilde{q}_{\troot})} \frac{1-\epsilon_{\troot}}{(1-q_{\troot}\epsilon_{\troot})^2}, 
\intertext{and, for $\troot>1$, }  
\fdf^{(\epsilon_{\troot})}
& = 
 \fdf^{(\tilde{p}_{\troot})}\frac{\partial \tilde{p}_{\troot}}{\partial \epsilon_{\troot}}
 + \fdf^{(\tilde{q}_{\troot})} \frac{\partial \tilde{q}_{\troot}}{\partial \epsilon_{\troot}}
= \fdf^{(\tilde{p}_{\troot})} \frac{1-q_{\troot}}{(1-q_{\troot}\epsilon_{\troot})^2}
	- \fdf^{(\tilde{q}_{\troot})} \frac{q_{\troot}(1-q_{\troot})}{(1-q_{\troot}\epsilon_{\troot})^2},
\end{align*}
as claimed in~\eqref{eq:fd.root}.

\item Now suppose that~$\ny$ is not the root. 
At any ancestor~$\nx_i$ with $i<d$, the distribution parameters of~$\ny$ 
affect the extinction probability~$\epsilon_{\nx_i}$. For 
a distribution parameter $\theta_{\ny}=p_{\ny}$, $\theta_{\ny}=q_{\ny}$, 
or $\theta_v=\epsilon_v$, 
\[
\frac{\partial \tilde{p}_{\nx_i}}{\partial \theta_{\ny}} = 
\frac{\partial \tilde{p}_{\nx_i}}{\partial \epsilon_{\nx_i}} %
\frac{\partial \epsilon_{\nx_i}}{\partial \tilde{p}_{\nx_{i+1}}}
\frac{\partial \tilde{p}_{\nx_{i+1}}}{\partial \theta_{\ny}}
\quad\text{and}\quad
\frac{\partial \tilde{q}_{\nx_i}}{\partial \theta_{\ny}} = 
\frac{\partial \tilde{q}_{\nx_i}}{\partial \epsilon_{\nx_i}} %
\frac{\partial \epsilon_{\nx_i}}{\partial \tilde{p}_{\nx_{i+1}}}
	\frac{\partial \tilde{p}_{\nx_{i+1}}}{\partial \theta_{\ny}},
\]
with 
\[
\frac{\partial \epsilon_{\nx_i}}{\partial \tilde{p}_{\nx_{i+1}}}
= \frac{\partial }{\partial \tilde{p}_{\nx_{i+1}}} \prod_{\nx_i\nz\in T} \tilde{p}_{\nz}
= \frac{\epsilon_{\nx_i}}{\tilde{p}_{\nx_{i+1}}}
\]
Let $\nx=\nx_{d-1}$ be the parent of $\ny=\nx_d$. 
Since $\frac{\partial \tilde{q}_{\ny}}{p_{\ny}}=0$, 
\begin{align*}
\fdf^{(p_{\ny})}
& = \Bigl(\fdf^{(\tilde{p}_{\ny})} 
	+ \fdf^{(\epsilon_{\nx})}
	 \frac{\partial \epsilon_{\nx}}{\partial \tilde{p}_{\ny}}\Bigr)
	 \frac{\partial \tilde{p}_{\ny}}{\partial p_{\ny}}
= \Bigl(\fdf^{(\tilde{p}_{\ny})}  
	+ \epsilon \fdf^{(\epsilon_{\nx})}\Bigr)
	\frac{1-\epsilon_{\ny}}{1-q_{\ny}\epsilon_{\ny}}
\end{align*}
with $\epsilon = \frac{\epsilon_{\nx}}{\tilde{p}_{\ny}}=\frac{\epsilon_{\nx_{d-1}}}{\tilde{p}_{\nx_d}}$.
The other two recurrences include $\fdf^{(\tilde{q}_{\ny})}$, as well:
\begin{align*} 
\fdf^{(q_{\ny})}
& = \Bigl(\fdf^{(\tilde{p}_{\ny})} 
	+ \fdf^{(\epsilon_{\nx})}
	 \frac{\partial \epsilon_{\nx}}{\partial \tilde{p}_{\ny}}\Bigr)
	 \frac{\partial \tilde{p}_{\ny}}{\partial q_{\ny}} 
	 +\fdf^{(\tilde{q}_{\ny})} 
		\frac{\partial \tilde{q}_{\ny}}{\partial q_{\ny}}
\\
& =  
		\Bigl(\fdf^{(\tilde{p}_{\ny})} 
		+ \epsilon \fdf^{(\epsilon_{\nx})}\Bigr)
		\frac{-(1-p_{\ny})\epsilon_{\ny}(1-\epsilon_{\ny})}{(1-q_{\ny}\epsilon_{\ny})^2}
	+ \fdf^{(\tilde{q}_{\ny})} 
		\frac{1-\epsilon_{\ny}}{(1-q_{\ny}\epsilon_{\ny})^2};	
\\
%
\fdf^{(\epsilon_{\ny})}
	& = \Bigl(\fdf^{(\tilde{p}_{\ny})}
	+ \epsilon \fdf^{(\epsilon_{\nx})}\Bigr)
	 \frac{\partial \tilde{p}_{\ny}}{\partial \epsilon_{\ny}} 
	 +\fdf^{(\tilde{q}_{\ny})} 
		\frac{\partial \tilde{q}_{\ny}}{\partial \epsilon_{\ny}}
\\
& = 
	\Bigl(\fdf^{(\tilde{p}_{\ny})}
	+ \epsilon \fdf^{(\epsilon_{\nx})}\Bigr)
		\frac{(1-p_{\ny})(1-q_{\ny})}{(1-q_{\ny}\epsilon_{\ny})^2}		
	- \fdf^{(\tilde{q}_{\ny})}
		\frac{q_{\ny}(1-q_{\ny})}{(1-q_{\ny}\epsilon_{\ny})^2},
\end{align*}
as shown in~\eqref{eq:fd.node}.
\end{enumerate}
\end{proof}

\subsection{No-duplication model}
\begin{proof}[Proof of Theorem~\ref{tm:empty.poisson}.]
Let $\nlik_{\nx},\slik_{\nx}$ denote the likelihoods 
for the empty profile: 
\begin{align*}
\nlik_{\nx}(n) & = \Probcmd{\forall \ny\in \mathcal{L}_{\nx} \colon \xi_{\ny}=0}{\xi_{\nx}=n}
\\
\slik_{\nx}(s) &= \Probcmd{\forall \ny \in \mathcal{L}_{\nx} \colon \xi_{\ny}=0}{\eta_{\nx}=s}.
\end{align*}
Let~$Q_{\nx}$ denote the product of $e^{-\tilde{r}_{\ny}}$
across all edges in the subtree of~$\nx$:
$Q_{\nx}=1$ at a leaf, and at an ancestral node~$\nx$
with children $\ny,\nz$
\[
Q_{\nx} = \bigl(Q_{\ny} e^{-\tilde{r}_{\ny}}\bigr)
	\bigl(Q_{\nz} e^{-\tilde{r}_{\nz}}\bigr).
\] 
We prove that 
for all nodes~$\nx$,
\[
\slik_{\nx}(s) = Q_{\nx} \times (\epsilon_{\nx})^s \times e^{-\tilde{r}_{\nx}}
\qquad
\nlik_{\nx}(n) = Q_{\nx} \times (\epsilon_{\nx})^n.
\]
(With $0^0=1$ and $0^n=0$ for $n>0$.)
In particular, at the root~$R$, $L(0)=\tilde{\slik}_R(0)$. 

We adjust the induction proof of Theorem~\ref{tm:empty}.  
At any node~$\nx$, with $r=r_{\nx}, \epsilon=\epsilon_{\nx}$ and $Q=Q_{\nx}$, 
by Equation~\eqref{eq:slik.rec.poisson},
\begin{align*}
\slik_{\nx}(s) & = \sum_{k=0}^{\infty} e^{-r}\frac{r^k}{k!} Q \epsilon^{s+k}
 = Q \epsilon^s e^{-r} \sum_{k=0}^{\infty} \frac{(r\epsilon)^k}{k!}
 \\ 
 & = Q \epsilon^s e^{-r(1-\epsilon)} = Q\epsilon^s e^{-\tilde{r}}
\end{align*}
with $\tilde{r}=r(1-\epsilon)$. 
The inductive case for $\nlik_{\nx}$ is adjusted:
\[
\nlik_{\nx} (n) = \prod_{\nx\ny\in T} \Bigl( Q_{\ny} e^{-\tilde{r}_{\ny}}\bigl(p_{\ny}+(1-p_{\ny})\epsilon_{\ny}\bigr)^n\Bigr)
 = Q_{\nx} (\epsilon_{\nx})^n.
\]
\end{proof}

\begin{proof}[Proof of Theorem~\ref{tm:fd.poisson}.]
Since $\tilde{p}_{\nx}=p_{\nx}+(1-p_{\nx})\epsilon_{\nx}$ and $\tilde{r}_{\nx}=r_{\nx}(1-\epsilon_{\nx})$,
\begin{align*}
\frac{\partial \tilde{p}_{\nx}}{\partial p_{\nx}} & = 1-\epsilon_{\nx}
& \frac{\partial \tilde{r}_{\nx}}{\partial r_{\nx}} & = 1-\epsilon_{\nx}
\\
\frac{\partial \tilde{p}_{\nx}}{\partial \epsilon_{\nx}} & = 1-p_{\nx}
&
\frac{\partial \tilde{r}_{\nx}}{\partial \epsilon_{\nx}} & = -r_{\nx}.
\end{align*}
The rest of the proof is based on applications of the chain rule as in the proof of Theorem~\ref{tm:fd}. 
\end{proof}

\subsection{Proof of Theorem~\ref{tm:time}}
\begin{proof}
Let $N=\sum_{\nz\in\mathcal{L}} n_{\nz}=m_R$ be the sum of copy numbers across the leaves. 
At an ancestral node~$\nx$, 
the calculations of $\tilde{\nlik}_{\nx}(\ell)$ for all $0\le \ell\le m_{\nx}$ 
and of $\solik_{\nx}(s)$ for all $0\le s\le m_{\nx}$ take
$(1+m_{\ny})(1+m_{\nz})$ iterations. 
Calculating $\tilde{\slik}(s)$ for all $0\le s\le m_{\nx}$ and 
$\nolik_{\nx}(\ell)$ for all $0\le \ell\le m_{\nx}$ 
is done in~$(1+m_{\nx})(2+m_{\nx})/2$ iterations.  
The total running time can be thus bounded 
asymptotically as $O\bigl(\sum_{\nx=1}^R m_{\nx}^2\bigr)$, 
or as $O(R)=O(L)$ if $N^2<R$, the number of nodes. 
Summing by the height of the nodes~$h(\nx)$, 
\begin{align*}
\sum_{\nx=1}^R m_{\nx}^2
& = \sum_{i=0}^{h-1}
	\sum_{\nx\colon h(\nx)=i} (m_{\nx})^2
	= \sum_{i=0}^{h-1} \sum_{\nx\colon h(\nx)=i} \Bigl(\sum_{\ny\in \mathcal{L}_{\nx}} n_{\ny}\Bigr)^2
\\
& \le \sum_{i=0}^{h-1} \Bigl(\sum_{\nx\colon h(\nx)=i}  \sum_{\ny\in \mathcal{L}_{\nx}} n_{\ny}\Bigr)^2
\\
& \le \sum_{i=0}^{h-1} \Bigl(\sum_{\ny\in \mathcal{L}} n_{\ny}\Bigr)^2
= h N^2.
\end{align*}
For the last inequality, note that if $h(\ny)=h(\nz)$ then their subtrees do not intersect
and $\mathcal{L}_{\ny}\cap \mathcal{L}_{\nz}=\emptyset$. 
\end{proof}

\subsection{An old algorithm for computing the profile likelihood}
The recursive algorithm  
of~\cite{phylobd.archaea,Count} 
for computing the 
profile likelihood uses the basic birth-death 
transitions from~\eqref{eq:transition.basic},
and arrives at a set of recurrences 
by combinatorial principles.  
We can infer the same method algebraically 
in the present framework --- the resulting formulas 
are not useful beyond serving up~$\tilde{C}$. 
As a warmup, we extract the recurrences for transition probabilities
from Theorem~\ref{tm:transition}. 
\begin{corollary}[Transition probability recurrences]\label{cor:transition.rec}
Let $\nx\ny\in T$ be any edge and 
$w(m\mid n) = \Probcmd{\xi_{\ny}=m}{\xi_{\nx}=n}$ 
denote the transition probabilities.

For $\lambda_{\ny}>0$, let $p=p_{\ny},q=q_{\ny}, \kappa=\kappa_{\ny}$ 
denote the applicable distribution parameters 
from Equation~\eqref{eq:param.probs}. 
Then 
\begin{align*}
w(m\mid 0) & = \binom{\kappa+m-1}{m} (1-q)^\kappa q^m\\
w(m\mid n) & = \begin{aligned}[t]
& q w(m\mid n-1) 
\\
+ & \{m>0\} (1-p-q) w(m-1\mid n-1) 
\\
+ &  \{m>0\}  q w(m-1\mid n).
\end{aligned} & \{ n>0\}
\end{align*}

For $\lambda=0$, let $p=p_{\ny}, r=r_{\ny}$ denote the applicable distribution parameters 
from Equation~\eqref{eq:param.probs}. 
Then 
\begin{align*}
w(m\mid 0) & = e^{-r}\frac{r^m}{m!}\\
w(m\mid n) & =  p w(m\mid n-1) 
+ \{m>0\} (1-p) w(m-1\mid n-1) & \{ n>0\}
\end{align*}
\end{corollary}
\begin{proof}\small
First, let $\lambda>0$. 
By Theorem~\ref{tm:transition}, the generating function for 
the transition probabilities is
\begin{align*}
G_n(z) & = \sum_{m=0}^\infty  w(m\mid n)  z^m \\
& = \Bigl(\frac{1-q}{1-q z}\Bigr)^{\kappa}\biggl(p +(1-p) \frac{(1-q)z}{1-q z}\biggr)^n
\\
& = \Bigl(\frac{1-q}{1-q z}\Bigr)^{\kappa}\biggl(\frac{p+z(1-p-q)}{1-q z}\biggr)^n.
%
\end{align*}
The generating function satisfies 
\[
G_n(z)\times (1-qz) = G_{n-1}(z) \times (p+z(1-p-q)).
\]
Noting that $z G_n(z) = \sum_{m=1}^\infty  w(m-1\mid n)  z^{m}$,
the equality of the coefficients implies 
that
\[
w(m\mid n)-q w(m-1\mid n) = p w(m\mid n-1) +(1-p-q) w(m-1\mid n-1),
\]
which is the Theorem's recurrence.  

For $\lambda=0$, the generating function is 
\[
G_n(z)=e^{-r(z-1)}\bigl(p+(1-p)z)^n,
\]
so $G_n(z)=G_{n-1}(z)\bigl(p+(1-p)z\bigr)$, giving the recurrence.  
\end{proof}

The profile likelihood algorithm of~\cite{phylobd.archaea,Count}  
combines the recurrences
of Theorem~\ref{tm:lik.rec.multi.right}, 
bypassing the explicit representation 
of conserved ancestral copies~$\tilde{\eta}$. 
Consider an ancestral node~$\nx$ with children $\ny_1,\dotsc, \ny_d$, 
and the step for computing $\tilde{\nlik}^{-(i-1)}_{\nx}(\ell)$ 
for some $0<i\le d$
from Equation~\eqref{eq:lik.rec.multi.right}.
Let~$\tilde{p}=\tilde{p}_{\ny_{i-1}}$, and 
$\epsilon=\epsilon_{\nx,-(i)}=\prod_{j=i}^{d} \tilde{p}_{\ny_j}$ : 
\begin{align*}
\tilde{\nlik}^{-(i-1)}_{\nx}(\ell)
& = \begin{aligned}[t]
	\sum_{j+k=\ell}\sum_{b=0}^{j} 
	& 
	\tilde{\slik}_{\ny}(j) \times \tilde{C}^{-(i)}_{\nx}(k+b) \\
	& \times 
	\binom{j+k}{j}
	\biggl(\frac{1-\tilde{p}}{1-\tilde{p}\epsilon}\biggr)^{j}
	\biggl(\frac{\tilde{p}-\tilde{p}\epsilon}{1-\tilde{p}\epsilon}\biggr)^{k}
	\binom{j}{b}(1-\epsilon)^b\epsilon^{j-b}.
	\end{aligned}
\\ & = \begin{aligned}[t]
	(1-\tilde{p}\epsilon)^{-\ell}
	\sum_{s+t=\ell} \tilde{C}^{-(i)}_{\nx}(t)
	& \times \binom{s+t}{s}(1-\epsilon)^s\epsilon^t \\
	& \times \sum_{b=0}^t \tilde{\slik}_{\ny_i}(s+b)\times  
 \underbrace{\binom{t}{b} (1-\tilde{p})^{s+b}\tilde{p}^{t-b}}_{\CB{=D_{\ny_i}(t,s)}}. 
 \end{aligned}
\end{align*}
\begin{subequations}\label{eq:algo.old}
The inner sum $D_{\ny_i}(t,s) = 
	\sum_{b=0}^t \tilde{\slik}_{\ny_i}(s+b)\times \binom{t}{b} (1-\tilde{p})^{s+b}\tilde{p}^{t-b}$
can be obtained by recursion for all $t>0$:
\begin{equation}\label{eq:algo.old.rec.edge}
D_{\ny_i}(t,s) 
 = D_{\ny_i}(t-1, s+1)+\tilde{p} D_{\ny_i}(t, s+1). 
\end{equation}
The starting values are $D_{\ny_i}(0,s)=(1-\tilde{p})^s\tilde{\slik}_{\ny_i}(s)$, 
which, by~\eqref{eq:lik.rec.edge}, 
further expands into 
\[
D_{\ny_i}(0,s)=
	\sum_{\ell\ge s} \tilde{\nlik}_{\ny_i}(\ell)\times 
	\underbrace{(1-\tilde{p})^s\binom{\kappa+\ell-1}{\ell-s} (1-\tilde{q})^{\kappa+s}\tilde{q}^{\ell-s}}_{\CB{=w^*(\ell\mid s)}}
\]
with $\tilde{q}=\tilde{q}_{\ny_i}$ and $\kappa=\kappa_{\ny_i}$. 
The transition weights are 
\[
w^{*}(\ell \mid s) = \Probcmd{\tilde{\xi}_{\ny_i}=\ell,\tilde{\eta}_{\ny_i}=s}{\tilde{\xi}_{\nx}=s},
\]
satisfying 
\begin{align}
w^*(\ell\mid 0) & = \binom{\kappa+\ell-1}{\ell-1} (1-\tilde{q})^{\kappa}\tilde{q}^\ell
	= h_{\ell}(t_{uv_i})\\
w^{*}(\ell\mid s) & = \{\ell>s\}\tilde{q} w^{*}(\ell-1\mid s) + (1-\tilde{p})(1-\tilde{q}) w^{*}(\ell-1 \mid s-1)
\qquad \{\ell>1\}
\end{align}
with the basic gain transitions $h_{\ell}$ from~\eqref{eq:transition.basic}. 
\begin{equation}
D_{\ny_i}(0,s)=
	\sum_{\ell\ge s} \tilde{\nlik}_{\ny_i}(\ell)\times w^*(\ell\mid s).
\end{equation}
We have thus the formulas relating $\tilde{C}$:
$\tilde{C}$: 
\begin{equation}\label{eq:algo.old.rec}
\tilde{\nlik}^{-(i-1)}_{\nx}(\ell) = \sum_{s+t=\ell}
	\tilde{\nlik}^{-i}_{\nx}(t)\times 
	D_{\ny_i}(t,s) \times 
	\binom{\ell}{s} (1-\epsilon)^s\epsilon^t 
\end{equation}
and 
\end{subequations}
The formulas of Equations~\eqref{eq:algo.old} form the basis of 
the algorithm reported in~\cite{phylobd.archaea}.
\begin{theorem}[An old algorithm for the profile likelihood]\label{tm:algo.old}
Let~$\Xi=\{\xi_{\nx}=n_{\nx}\mid \nx\in\mathcal{L}\}$ be an arbitrary profile 
with an average of $\bar{n}=\frac1{L}\sum_{\nx\in\mathcal{L}} n_{\nx}$ 
copies across $L=|\mathcal{L}|$ leaves. 
The profile likelihood  
can be computed in a postorder traversal of the phylogeny, 
using the formulas of Equations~\eqref{eq:algo.old.rec.edge}--\eqref{eq:algo.old.rec}, 
as
\[
L(\Xi) = \sum_{\ell\ge 0} \tilde{\nlik}_R(\ell) \binom{\kappa_{R}+\ell-1}{\ell} (1-\tilde{q}_R)^{\kappa_R} (\tilde{q}_R)^\ell.
\]
with at the root~$R$.  
The computations take $O(hL( L\bar{n}^2+1))$ time if the phylogeny 
height is~$h$.
\end{theorem}

\section*{Acknowledgments}
This research did not receive any specific grant whatsoever from funding agencies in the public, commercial, or not-for-profit sectors.


\begin{thebibliography}{10}
\expandafter\ifx\csname url\endcsname\relax
  \def\url#1{\texttt{#1}}\fi
\expandafter\ifx\csname urlprefix\endcsname\relax\def\urlprefix{URL }\fi
\expandafter\ifx\csname href\endcsname\relax
  \def\href#1#2{#2} \def\path#1{#1}\fi

\bibitem{Fitch.homology}
W.~M. Fitch, Homology a personal view on some of the problems, Trends in
  Genetics 16~(5) (2000) 227--231.

\bibitem{RAST}
R.~K. Aziz, et~al., The {RAST} server: {R}apid {A}nnotations using {S}ubsystems
  {T}echnology, BMC Genomics 9 (2008) 75.
\newblock \href {https://doi.org/10.1186/1471-2164-9-75}
  {\path{doi:10.1186/1471-2164-9-75}}.

\bibitem{phylobd.archaea}
M.~Cs\H{u}r\"os, I.~Mikl\'os, Streamlining and large ancestral genomes in
  {A}rchaea inferred with a phylogenetic birth-and-death model, Molecular
  Biology and Evolution 26~(9) (2009) 2087--2095.
\newblock \href {https://doi.org/10.1093/molbev/msp123}
  {\path{doi:10.1093/molbev/msp123}}.

\bibitem{DeyMeyer}
G.~Dey, T.~Meyer, Phylogenetic profiling for probing the modular architecture
  of the human genome, Cell Systems 1 (2015) 106--115.
\newblock \href {https://doi.org/10.1016/j.cels.2015.08.006}
  {\path{doi:10.1016/j.cels.2015.08.006}}.

\bibitem{Nye.genecontent}
T.~M.~W. Nye, Modelling the evolution of multi-gene families, Statistical
  Methods in Medical Research 18 (2009) 487--504.
\newblock \href {https://doi.org/10.1177/0962280208099450}
  {\path{doi:10.1177/0962280208099450}}.

\bibitem{Takacs}
L.~Tak\'acs, Introduction to the Theory of Queues, Oxford University Press, New
  York, 1962.

\bibitem{Kendall}
D.~G. Kendall, Stochastic processes and population growth, Journal of the Royal
  Statistical Society Series B 11~(2) (1949) 230--282.

\bibitem{Felsenstein.ML.pruning}
J.~Felsenstein, Maximum likelihood and minimum-steps methods for estimating
  evolutionary trees from data on discrete characters, Systematic Zoology
  22~(3) (1973) 240--249.

\bibitem{Hahn2005}
M.~W. Hahn, T.~De~Bie, J.~E. Stajich, C.~Nguyen, N.~Cristianini, Estimating the
  tempo and mode of gene family evolution from comparative genomic data, Genome
  Research 15 (2005) 1153--1160.
\newblock \href {https://doi.org/10.1101/gr.3567505}
  {\path{doi:10.1101/gr.3567505}}.

\bibitem{IwasakiTakagi}
W.~Iwasaki, T.~Takagi, Reconstruction of highly heterogeneous gene-content
  evolution across the three domains of life, Bioinformatics 23~(13) (2007)
  i230--i239.
\newblock \href {https://doi.org/10.1093/bioinformatics/btm165}
  {\path{doi:10.1093/bioinformatics/btm165}}.

\bibitem{Istvan.genecontent}
M.~Cs{\H u}r\"os, I.~Mikl\'os, A probabilistic model for gene content evolution
  with duplication, loss, and horizontal transfer, Springer Lecture Notes in
  Bioinformatics 3909 (2006) 206--220, proc.\ Tenth Annual International
  Conference on Research in Computational Molecular Biology (RECOMB).
\newblock \href {https://doi.org/10.1007/11732990_18}
  {\path{doi:10.1007/11732990_18}}.

\bibitem{Count}
M.~Cs\H{u}r\"os, Count: evolutionary analysis of phylogenetic profiles with
  parsimony and likelihood, Bioinformatics 26~(15) (2010) 1910--1912.
\newblock \href {https://doi.org/10.1093/bioinformatics/btq315}
  {\path{doi:10.1093/bioinformatics/btq315}}.

\bibitem{SelaWolfKoonin}
I.~Sela, Y.~I. Wolf, E.~V. Koonin, Theory of prokaryotic evolution, Proceedings
  of the National Academy of Sciences of the USA 113 (2016) 11399--11407.
\newblock \href {https://doi.org/10.1073/pnas.1614083113}
  {\path{doi:10.1073/pnas.1614083113}}.

\bibitem{TKF}
J.~L. Thorne, H.~Kishino, J.~Felsenstein, An evolutionary model for maximum
  likelihod alignment of {DNA} sequences, Journal of Molecular Evolution 33
  (1991) 114--124.

\bibitem{KarlinMcgregor}
S.~Karlin, J.~McGregor, Linear growth, birth, and death processes, Journal of
  Mathematics and Mechanics 7~(4) (1958) 643--662.

\bibitem{Felsenstein.restml}
J.~Felsenstein, Phylogenies from restriction sites, a maximum likelihood
  approach, Evolution 46 (1992) 159--173.

\end{thebibliography}

%
%

\clearpage
\tableofcontents

\end{document}